\documentclass[sigconf, nonacm]{acmart}

\newcommand\vldbdoi{XX.XX/XXX.XX}
\newcommand\vldbpages{XXX-XXX}
\newcommand\vldbvolume{14}
\newcommand\vldbissue{1}
\newcommand\vldbyear{2020}
\newcommand\vldbauthors{\authors}
\newcommand\vldbtitle{\shorttitle} 
\newcommand{\revnote}[1]{%
  \marginnote{\tiny\bfseries\color{blue!70!black}#1}%
}
\newcommand\vldbavailabilityurl{URL_TO_YOUR_ARTIFACTS}
\newcommand\vldbpagestyle{plain} 
\usepackage[T1]{fontenc}
\usepackage[utf8]{inputenc}
\usepackage{microtype}
\usepackage{xcolor}
\usepackage{marginnote}
\newcommand{\rev}[1]{\textcolor{blue}{#1}}

\usepackage{url}
\usepackage{multirow}

 \usepackage{amsmath,amssymb,amsfonts}

\usepackage{booktabs}
\usepackage{enumitem}

\usepackage{graphicx}
\usepackage{caption}
\usepackage{subcaption}

\usepackage{weiwAlgorithm}

\usepackage{float}
\usepackage{wrapfig}
\usepackage{balance}

\theoremstyle{definition}
\newtheorem{definition}{Definition}

\theoremstyle{plain}
\newtheorem{theorem}{Theorem}
\newtheorem{lemma}{Lemma}

\theoremstyle{remark}

\begin{document}
\title{Sparse Neighborhood Graph-Based Approximate Nearest Neighbor Search Revisited: Theoretical Analysis and Optimization}

\author{Xinran Ma}
\authornote{Both authors contributed equally to this work.}
\affiliation{%
  \institution{Academy of Mathematics and Systems Science, Chinese Academy of Sciences}
}
\email{maxinran22@mails.ucas.ac.cn}

\author{Zhaoqi Zhou}
\authornotemark[1]
\orcid{}
\affiliation{%
  \institution{Huawei Technologies Co., Ltd.}
}
\email{zhouzhaoqi1@huawei.com}

\author{Chuan Zhou}
\orcid{}
\affiliation{%
  \institution{Academy of Mathematics and Systems Science, Chinese Academy of Sciences}
}
\email{zhouchuan@amss.ac.cn}

\author{Zaijiu Shang}
\affiliation{%
  \institution{Shanghai Institute for Mathematics and Interdisciplinary Sciences}
}
\email{zaijiu@simis.cn}

\author{Guoliang Li}
\affiliation{%
  \institution{Tsinghua University}
}
\email{liguoliang@tsinghua.edu.cn}

\author{Zhiming Ma}
\affiliation{%
  \institution{Academy of Mathematics and Systems Science, Chinese Academy of Sciences}
}
\email{mazm@amt.ac.cn}





\begin{abstract}
 Graph-based approaches to approximate nearest neighbor search (ANNS) enable fast, high-recall retrieval on billion-scale vector datasets. Among them, the Sparse Neighborhood Graph (SNG) is widely used due to its strong search performance. However, the lack of theoretical understanding of SNG leads to expensive tuning of the truncation parameter that controls graph sparsification. In this work, we present OPT-SNG, a principled framework for analyzing and optimizing SNG construction. We introduce a martingale-based model of the pruning process that characterizes the stochastic evolution of candidate sets during graph construction. Using this framework, we prove that SNG has a maximum out-degree of \(O(n^{2/3+\epsilon})\), where \(\epsilon>0\) is an arbitrarily small constant, and an expected search path length of \(O(\log n)\). Building on these insights, we derive a closed-form rule for selecting the optimal truncation parameter \(R\), thereby eliminating the need for costly parameter sweeping. Extensive experiments on real-world datasets demonstrate that OPT-SNG achieves an average \(5.9\times\) speedup in index construction time, with peak improvements reaching \(15.4\times\), while consistently maintaining or improving search performance.
\end{abstract}

\maketitle

\pagestyle{\vldbpagestyle}
\begingroup\small\noindent\raggedright\textbf{PVLDB Reference Format:}\\
\vldbauthors. \vldbtitle. PVLDB, \vldbvolume(\vldbissue): \vldbpages, \vldbyear.\\
\href{https://doi.org/\vldbdoi}{doi:\vldbdoi}
\endgroup
\begingroup
\renewcommand\thefootnote{}\footnote{\noindent
This work is licensed under the Creative Commons BY-NC-ND 4.0 International License. Visit \url{https://creativecommons.org/licenses/by-nc-nd/4.0/} to view a copy of this license. For any use beyond those covered by this license, obtain permission by emailing \href{mailto:info@vldb.org}{info@vldb.org}. Copyright is held by the owner/author(s). Publication rights licensed to the VLDB Endowment. \\
\raggedright Proceedings of the VLDB Endowment, Vol. \vldbvolume, No. \vldbissue\ %
ISSN 2150-8097. \\
\href{https://doi.org/\vldbdoi}{doi:\vldbdoi} \\
}\addtocounter{footnote}{-1}\endgroup

\ifdefempty{\vldbavailabilityurl}{}{
\vspace{.3cm}
\begingroup\small\noindent\raggedright\textbf{PVLDB Artifact Availability:}\\
The source code, data, and/or other artifacts have been made available at \url{https://github.com/maxila320/OPT-SNG}.
\endgroup
}
\section{Introduction}
 Nearest Neighbor Search (NNS) is a fundamental problem in computer science~\cite{KNN1,KNN2,knn3,KNN4,knn5,knn6} and has become a fundamental research topic across various fields such as information retrieval~\cite{IRTV,image,retr2}, data mining~\cite{Dataminingsig2017,datamining2,datamining3}, and large language models~\cite{RAG,RAG2,ICDEML}. Formally, given a dataset \(P = \{p_1, \ldots, p_n\} \subset \mathbb{R}^d\) of cardinality \(n\) and query point \(q \in \mathbb{R}^d\), the NNS problem requires finding the \(k\) closest points to \(q\) under an appropriate distance metric \(\|\cdot\|\). The computational complexity of exact NNS grows prohibitively with both dimensionality \(d\) due to the well-documented curse of dimensionality~\cite{Curse,curse2,curse3,curse4,curse5} and the expanding dataset size \(n\). This fundamental limitation has motivated the development of Approximate Nearest Neighbor Search (ANNS) algorithms, which achieve significant efficiency improvements while sacrificing minimal accuracy.

 Modern ANNS approaches achieve remarkable search performance and fall into four categories: tree-based approaches~\cite{tree1,KD-TREE}, hashing-based approaches~\cite{LSHarshing,hash2,hash3,hash4}, vector quantization-based approaches~\cite{quantti1,RabitQ}, and graph-based approaches~\cite{Graph4, Graph3, Graph1,UNI1}. Among these methods, graph-based approaches have demonstrated particular success in practical applications~\cite{Graph2024VLDB, Graph1}, as they effectively capture and exploit local neighborhood structures to enable navigation in high-dimensional spaces. A crucial observation is that despite the diversity of graph construction strategies—which vary in neighbor selection and edge formation rules—most state-of-the-art graph-based ANNS systems share a unified query processing paradigm. The graph-based ANNS framework operates through two complementary phases: (1) an index construction phase that builds a proximity graph preserving neighborhood relationships through various graph refinement strategies, and (2) a query processing phase that employs the \textit{GreedySearch} algorithm~\cite{NSG19,NEURIPS2019,HNSW2020} to iteratively traverse the graph along decreasing-distance paths until convergence.


 
  Sparse Neighborhood Graph (SNG) is a widely used variant designed to further reduce graph density while preserving practical search quality~\cite{Arya93}. It has been widely adopted in several state-of-the-art systems, including DiskANN~\cite{NEURIPS2019}, HNSW~\cite{HNSW2020} and NSG~\cite{NSG19}. The worst-case search analysis of SNG by Indyk et al.~\cite{indyk2023worst} provides a framework based on aspect ratios\footnote{The aspect ratio is the ratio of the maximum to minimum distance between any pair of points in the dataset.}. Complementary to this, a developing and expected cardinality-based framework is both meaningful and necessary for understanding SNG behavior under random or empirical data distributions.
 
The construction process of SNG typically involves two main phases: 
(1) Initial graph generation through a random or approximate k-NN graph, and (2) Graph refinement through iterative pruning\cite{revisiting2025}. The pruning phase is the computational bottleneck, with \(O(n)\) worst-case complexity per node, as it requires scanning all candidate neighbors and evaluating geometric conditions that ensure the graph supports fast convergence of the \textit{GreedySearch} algorithm.
To reduce index construction time, a common practice in SNG construction is to impose a maximum node degree via a truncation parameter $R$. However, existing methods typically determine \(R\) through parameter sweeping (or a parameter study)~\cite{NEURIPS2019}, commonly implemented as a binary search guided by the downstream search performance of the constructed graph. This tuning process is computationally expensive, as each iteration requires constructing and evaluating a full index. These limitations highlight the need for a theoretically grounded approach to selecting \(R\)—one that ensures high search performance while reducing tuning overhead. 

To address both the theoretical and practical challenges in SNG-based approximate nearest neighbor search, we propose \textbf{OPT-SNG}, a 
martingale-based probabilistic framework for analyzing and 
optimizing SNG construction. In this framework, 
we characterize the SNG construction process as a discrete stochastic 
process driven by successive pruning operations. By exploiting the 
martingale structure naturally arising from the evolution of candidate 
sets, we establish that the expected trajectory can be characterized by 
an associated ordinary differential equation (ODE). Furthermore, we prove 
that the out-degree of each node corresponds precisely to the stopping 
time of this stochastic process. Within this framework, we derive bounds on two fundamental structural properties of SNGs—the maximum out-degree and the expected search path length. Building upon these theoretical results, we further quantify index construction complexity, which serves as a principled foundation for optimizing the truncation parameter \(R\) in SNG-based ANNS systems.

We conclude our contributions as follows.
\begin{itemize}[leftmargin=*]
    \item We introduce OPT-SNG, the first martingale-based probabilistic model of SNG pruning dynamics. By formulating the pruning process as a submartingale with respect to 
    a natural filtration, we transform graph refinement into an analyzable 
    probabilistic system. This enables the application of concentration 
    inequalities and differential equation methods to study graph evolution.
\item  We establish a degree bound of 
    $O(n^{2/3+\epsilon})$ for any small constant $\epsilon > 0$, substantially improving over the previous coarse 
    $O(n)$ bound from \cite{NEURIPS2019}. This result highlights the 
    inherent sparsity of SNGs and demonstrates that graph structure can be 
    controlled through proper pruning dynamics analysis. 
\item We prove that GreedySearch on SNGs converges in $O(\log n)$ expected steps. This 
    analytical guarantee validates the intuition behind SNG's empirical 
    efficiency and provides theoretical justification.
    \item Using the above theoretical foundations, we analyze the asymptotic complexity of the truncated SNG construction process and derive a closed-form expression for the optimal truncation parameter \(R\). This yields a principled and automated parameter selection method, avoiding costly parameter sweeping over \(R\). Empirically, our method achieves 
    comparable or superior search performance while reducing index construction 
    time by an average of $5.9\times$ across datasets, with peak speedups 
    reaching $15.4\times$.
\end{itemize}    
All detailed proofs of lemmas and theorems are provided in the full version~\cite{Ma2025} 
of this work on our GitHub repository.


The remainder of this paper is organized as follows. Section~\ref{sec:preliminaries} reviews the necessary background on SNG-based ANNS, formally defines the problem setting, and introduces the key algorithms considered in this study. In Section~\ref{sec:theoretical_analysis}, we present our probabilistic modeling framework that captures the stochastic dynamics of the pruning process. Sections~\ref{sec:degree} and~\ref{sec:path} derive the main theoretical bounds on graph degree and search path length, respectively. 
Section~\ref{sec:algorithm_optimization} describes our optimized construction algorithm and the principled method for truncation parameter selection. Experimental results validating our theoretical predictions and demonstrating practical speedups are reported in Section~\ref{sec:experiments}. We survey the related works and position our contributions within the broader literature in Section~\ref{sec:related_work}, and Section~\ref{sec:conclusion} concludes with a summary and directions for future research.

\section{\(k-\)ANN problem and Sparse Neighborhood Graphs}
\label{sec:preliminaries}
In this section, we formalize the approximate nearest search problem formulation, review existing SNG-based ANNS construction, and summarize the key algorithmic components relevant to our theoretical analysis. Given a dataset $P = \{p_1, \ldots, p_n\} \subset \mathbb{R}^d$ of $n$ points 
in $d$-dimensional Euclidean space with distance metric $\|\cdot\|$, we consider a query point \(q \in \mathbb{R}^d\). 
See Table~\ref{tab:notation} for a summary of key notations used throughout.

\subsection{Problem Formulation}
We formally define the \(k\)-approximate nearest neighbor search problem as follows.

\begin{definition}[$k$-Approximate Nearest Neighbor Search]
\label{def:k-anns}
Given a dataset \(P \subset \mathbb{R}^d\), a query point \(q \in \mathbb{R}^d\), and an integer \(k \geq 1\), the \(k\)-approximate nearest neighbor search (K-ANNS) problem aims to return a set \(A \subseteq P\) of \(k\) data points that are closest to \(q\), subject to a desired trade-off between recall and query latency.
\end{definition}


Graph-based ANNS methods address this problem by modeling the dataset \(P\) as the vertex set of a graph and constructing edges between data points to facilitate efficient search. Typically, the construction process begins with an initial graph, such as a random graph or a \(k\)-nearest neighbor graph, and then refines this structure through iterative edge pruning. Formally, these methods construct a graph \(G = (P, E)\), where each vertex corresponds to a data point, and the edge set \(E\) is designed to preserve neighborhood proximity, enabling fast approximate nearest neighbor search via local graph traversal.



\subsection{\textit{GreedySearch} Algorithm}
Despite differences in index construction, many graph-based ANNS methods share
a common query-processing phase: a greedy best-first traversal known as
\textit{GreedySearch}.
Representative examples include KGraph~\cite{Kgraph}, DiskANN~\cite{NEURIPS2019},
NSG~\cite{NSG19}, HNSW~\cite{HNSW2020}, DPG~\cite{DPG}, and SPTAG~\cite{SPANN2021}.
The procedure is described in Algorithm~\ref{alg:greedy_search}.

Starting from an entry point \(s\), the algorithm maintains (1) a visited set and (2) a priority queue of candidate nodes ordered by distance to the query \(q\). At each iteration, the closest unexpanded vertex \(\hat{p}\) is extracted from the priority queue. Its outgoing neighbors are then explored and added to the priority queue if they have not been visited, and simultaneously, \(\hat{p}\) is added to the visited set. The process continues until either the search width reaches the predefined budget \(L\) or the priority queue becomes empty. The parameter \(L\) controls the search width and governs the recall and latency trade-off: larger values of \(L\) explore more candidate nodes, improving recall at the cost of increased search time. 
\begin{algorithm}[t]
\small
    \SetVline 
    \SetFuncSty{textsf}
    \SetArgSty{textsf}
\caption{\texttt{GreedySearch}($G, q, s, k, L$)}
\label{alg:greedy_search}
\Input{graph $G=(P,E)$, query point $q$, entry point $s$, $k$ for top-$k$, search width $L$}
\Output{top-$k$ approximate nearest neighbors of $q$}

\State{$Visited \leftarrow \{s\}$}
\State{$Queue \leftarrow$ priority queue initialized with $s$}

\While{$|Visited| < L$ \textbf{and} $Queue \neq \varnothing$}
{
    \State{$\hat{p} \leftarrow$ extract-min$(Queue)$}
    \tcc{closest unexpanded point}
    \For{each {$p' \in N_{\text{out}}(\hat{p})$}}
    {
        \If{$p' \notin Visited$}
        {
            \State{$Visited \leftarrow Visited \cup \{p'\}$}
            \State{insert $p'$ into $Queue$}
        }
    }
}
\Return{the $k$ points in $Visited$ closest to $q$}
\end{algorithm}

Maintaining and accelerating the convergence of \textit{GreedySearch} is a central objective in graph construction. In particular, ensuring graph connectivity while preserving monotonic search paths is both subtle and essential for guaranteeing fast convergence and high recall. The Sparse Neighborhood Graph (SNG)~\cite{Arya93} has been shown to 
effectively achieve these objectives by iteratively selecting nearest 
neighbors and pruning redundant candidates~\cite{indyk2023worst,Haya24}.

\subsection{Sparse Neighborhood Graph Construction}
\label{sec:2.3}
A core operation in SNG construction is the pruning procedure, which 
selectively retains informative neighbors while aggressively eliminating 
redundant connections. This procedure is a fundamental building block in 
several state-of-the-art graph-based ANNS systems, including DiskANN~\cite{NEURIPS2019}, 
NSG~\cite{NSG19}, and HNSW~\cite{HNSW2020}.
\begin{algorithm}[t]
\small
\SetVline
\SetFuncSty{textsf}
\SetArgSty{textsf}
\caption{\texttt{SNG-Prune}$(p, S, \alpha, R)$}
\label{alg:robust_prune}

\KwIn{
point $p$,
candidate set $S$,
pruning parameter $\alpha \ge 1$,
degree bound $R$
}
\KwOut{
pruned neighbor set $N_{\text{out}}(p)$ with $|N_{\text{out}}(p)| \le R$
}

$N_{\text{out}}(p) \gets \emptyset$\;
Sort $S$ in ascending order of $\|p-\cdot\|$\;

\While{$S \neq \emptyset$}{
    $p^* \gets \arg\min_{p' \in S} \|p-p'\|$\;
    $N_{\text{out}}(p) \gets N_{\text{out}}(p) \cup \{p^*\}$\;
    Remove $p^*$ from $S$\;

    \If{$|N_{\text{out}}(p)| = R$}{
        \textbf{break}\;
    }

    \ForEach{$p' \in S$}{
        \If{$\|p-p'\| \ge \alpha \cdot \|p^*-p'\|$}{
            Remove $p'$ from $S$\;
        }
    }
}
\Return{$N_{\text{out}}(p)$}\;
\end{algorithm}

 The algorithm proceeds iteratively: in each iteration, it selects the 
nearest point $p^*$ from the candidate set $S$, adds $p^*$ to the neighbor 
set $N_{\text{out}}(p)$, and then prunes candidates from $S$. Specifically, 
a candidate $p'$ is removed if it satisfies the blocking condition  
\(
\|p-p'\| \ge \alpha \cdot \|p^*-p'\|.
\)
This condition captures the intuition that $p'$ is blocked by the 
recently selected neighbor $p^*$—that is, accessing $p'$ through the 
intermediate neighbor $p^*$ is at least as efficient as a direct connection. 
Consequently, the direct edge $p \to p'$ is redundant and can be safely removed 
without compromising graph navigability.

Two parameters control the pruning behavior:

\noindent \textbf{Pruning Aggressiveness Parameter} $\alpha$: In the 
non-relaxed setting ($\alpha = 1$), the blocking condition is symmetric and 
purely geometric. However, to avoid pathological configurations (e.g., collinear 
points), the Vamana algorithm~\cite{NEURIPS2019} introduces $\alpha > 1$, which 
relaxes the blocking condition and ensures bounded graph diameter. In practice, 
$\alpha$ is typically fixed in the range $1 \leq \alpha \leq 2$~\cite{OOD} and 
does not require extensive tuning, as the SNG structure remains robust across 
this range.

\noindent \textbf{Degree Truncation Parameter} $R$: To control graph density 
and construction cost, a maximum degree limit $R$ is imposed. The algorithm 
terminates after exactly $R$ neighbors have been selected, bounding the out-degree 
of each node. This truncation ensures that construction remains tractable while 
allowing a trade-off between graph quality and computational efficiency.

A crucial property of SNG-based construction is that it preserves \emph{monotonic searchability}~\cite{NSG19}. Specifically, for any point \(p\), there exists at least one monotonic path consisting of edges that lead to progressively closer points, eventually reaching the nearest neighbor of \(p\). This property is essential for the effectiveness of \textit{GreedySearch}, as it guarantees that each expansion step can make consistent progress toward the target.

Despite its empirical success across multiple state-of-the-art systems (DiskANN, NSG, HNSW), the theoretical understanding of SNG construction remains incomplete. The introduction of relaxation parameter $\alpha$ substantially complicates the analysis: the blocking condition becomes asymmetric, and the geometric properties become more intricate. While Indyk et al.~\cite{indyk2023worst} provide worst-case analysis based on aspect ratios and Prokhorenkova and Shekhovtsov \cite{Haya24} establish bounds under specific distributional assumptions, a comprehensive expected-case, cardinality-based framework characterizing SNG behavior has been lacking.

In particular, the truncation parameter $R$ remains a critical tuning knob. Existing practice relies on expensive parameter sweeping: practitioners repeatedly construct indices with different $R$ values, evaluate their search performance through binary or grid search, and select the value that optimizes the recall-latency trade-off. This iterative process is computationally prohibitive for large-scale datasets and obscures the fundamental principles governing the relationship between $R$, graph structure, and search efficiency.

Our OPT-SNG framework addresses these gaps by providing a principled, theoretically grounded approach to analyzing SNG dynamics and optimizing 
the truncation parameter $R$.
\section{OPT-SNG Probabilistic Framework Setup}
\label{sec:theoretical_analysis}
\begin{figure}
    \centering
\includegraphics[width=1\linewidth]{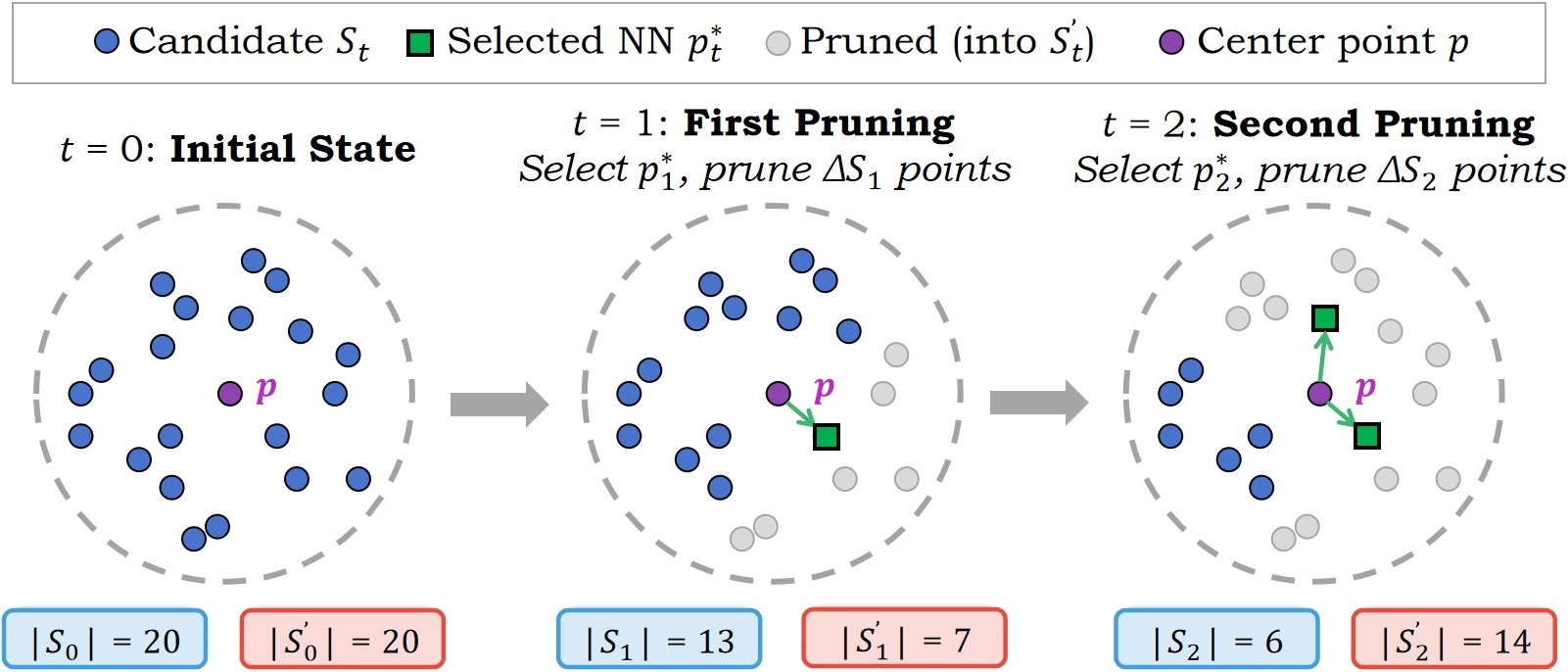}
    \caption{Illustration of the SNG process for a single point.}
    \label{fig:set_example}
\end{figure}
  
\textbf{Stochastic Modeling of the Pruning Process.} To analyze the construction of SNG from a dynamic perspective, we establish the following probabilistic framework that models the evolution of the pruning process. Consider a fixed point \( p \in P \) being processed. Let \( t \in \{0,1, 2, \dots\} \) denote the iteration index, where each iteration is associated with two dynamically evolving sets.

\begin{definition}[Candidate Set and Processed Set]
At iteration \(t\), the \emph{candidate set} is denoted by \( S_t \). It is initialized as
\(
S_0 = P \setminus \{p\},
\)
meaning that, initially, the point \(p\) is allowed to connect to all other points in the dataset. At each iteration, \(S_t\) consists of the points that have not yet been pruned and have not been selected as neighbors.
The corresponding \emph{processed set} is defined as
\(
S'_t = P \setminus (S_t \cup \{p\}),
\)
which contains all points that have either been pruned or already selected as nearest neighbors in previous iterations.
\end{definition}
In the non-truncated setting, the pruning process terminates when \( S_t = \emptyset \), or equivalently when \( S'_t = P \setminus \{p\} \).  Figure~\ref{fig:set_example} provides a schematic illustration of this construction process. Initially ($t=0$), all other points belong to the candidate set $S_0$. At iteration $t=1$, the nearest neighbor $p_1^*$ is selected, and a subset of points $\Delta S_1$ is pruned and moved into the processed set $S'_1$. 
This process repeats at subsequent iterations: at iteration $t$, a new nearest neighbor $p_t^*$ is selected from $S_{t-1}$, additional points are pruned, and the candidate set shrinks accordingly.

 Due to the randomness of point distributions in high-dimensional space, the number of pruned points at each iteration is itself random, which induces randomness in the resulting graph structure. Consequently, the pruning procedure can be naturally viewed as a stochastic process. The correspondence between iteration indices and the evolution of \( S_t \) and \( S'_t \) allows us to formalize the construction using tools from probability theory.

\textbf{Probability Space for SNG construction.} We introduce a probability space \( (\Omega, \mathcal{F}, \mathbb{P}) \), where \(\Omega\) denotes the sample space containing all possible sequences of nearest neighbor selections and pruning outcomes during the overall SNG construction. The \(\sigma\)-algebra  \(\mathcal{F}\) captures all measurable events on \(\Omega\) , and \(\mathbb{P}\) is the associated probability measure. We further define the \emph{natural filtration} \( \{\mathcal{F}_t\}_{t \geq 1} \), representing the information available up to iteration \( t \). Specifically, 
\[
\mathcal{F}_t = \sigma\left((p^*_1, S_1), (p^*_2,S_2), \dots, (p_t^*,S_t)\right)
\]
where \( p^*_i = \arg\min_{s\in S_i} \|p-s\| \) is the nearest neighbor selected at iteration \( i \). The candidate set is updated according to the pruning rule \( S_{t+1} = S_t \setminus (\{ p' \in S_t : \|p- p'\|\geq \alpha \cdot \|p^*- p'\| \} \cup \{p^*_{t+1}\} )\) which removes both the selected nearest neighbor and all candidates blocked by it. This formulation captures the pruning process as a sequence of progressively revealed decisions.

\textbf{Martingale Perspective.} We now formalize the dynamic trend of the construction process.  Intuitively, at each iteration, the processed set strictly grows while the candidate set shrinks. This monotonic behavior can be rigorously characterized using martingale theory. 

\begin{definition}[Martingale]
A stochastic process \( \{X_t\}_{t \geq 0} \) is a \emph{martingale} with respect to a filtration \( \{\mathcal{F}_t\} \) if
 \[\ \mathbb{E}[|X_t|] < \infty \quad and \quad \mathbb{E}[X_{t+1} | \mathcal{F}_t] = X_t \]
for all \( t \geq 0 \).
It is a \textit{supermartingale} if  \( \mathbb{E}[X_{t+1} | \mathcal{F}_t] \leq X_t \), and a \textit{submartingale} if \( \mathbb{E}[X_{t+1} | \mathcal{F}_t] \geq X_t \). 
\end{definition}

\begin{table}[t]
\centering
\caption{Summary of Notation}
\label{tab:notation}
\begin{tabular}{cc}
\toprule
\textbf{Symbol} & \textbf{Description} \\
\midrule
\( P \) & Dataset of \( n \) points in \( \mathbb{R}^d \) \\
\( S_t \) & Candidate set at iteration \( t \) \\
\( S'_t \) & Processed set at iteration \( t \): \( S'_t = P \setminus (S_t \cup \{p\}) \) \\
\( p^*_t \) & Nearest neighbor of \( p \) in \( S_t \) at iteration \( t \) \\
\( \mathcal{F}_t \) & Natural filtration up to iteration \( t \) \\
\( \Delta S_t \) & Number of points pruned in iteration \( t \) \\
\( \pi_t \) & Pruning probability at iteration \( t \) \\
\( R \) & Truncation parameter (maximum degree limit) \\
\( M \) & Number of processed points at a given level \\
\(\alpha\)& Relaxation parameter\\
\( t \) & Iteration index (time step) \\
\bottomrule
\end{tabular}
\end{table}

In the context of SNG construction, the size of the candidate set \( |S_t| \) forms a supermartingale, while the size of the processed set \( |S'_t| \) forms a submartingale with respect to the natural filtration \( \{\mathcal{F}_t\} \).  This reflects the fact that, in expectation, candidates are eliminated while processed points accumulate over time. We focus on $|S'_t|$ rather than $|S_t|$ because, as complementary sets, 
$|S'_t|$'s monotonic growth is suited to our analytical framework.

From the SNG pruning rule and the probability framework, we summarize the following properties of processed sets \(S'_t\), which will facilitate the analysis of the construction process. These properties follow directly from our previous results and do not require separate proof.
\begin{lemma}[Dynamic Property of \(S'_t\)]
\label{lemma:prop}
Given a dataset \( P \subset \mathbb{R}^d \) consisting of \( n \) points, the SNG construction of \(p \in P\) satisfies the properties as follows:\\
\begin{itemize}[leftmargin=*]
    \item \textbf{Monotonicity}: The processed sets form a nested sequence \(\emptyset=S_0'\subset S'_1 \subset S'_2 \subset \cdots \subset S'_T \), with a submartingale \( |S'_t| \).
    \item \textbf{Termination}: The final processed set includes \(n-1\) points, excluding only the current point \( p \).
    \item \textbf{Degree property}: Out-degree of \(p\) equals the number of iterations, as each iteration adds exactly one directed edge from \(p\) to its nearest neighbor in the current candidate set.
\end{itemize}
\end{lemma}

Based on this framework, we introduce the following metric to measure the progression of graph construction.

\begin{definition}[\( M \)-\( t \) Level]
The SNG construction is said to reach \textit{the \( M \)-\( t \) level} at iteration \( t \) if \(t\) is the smallest iteration that \( |S'_t| \geq M \). This represents the \textit{first passage time} at which at least \( M \) points have been processed.
\end{definition}

The $M$-$t$ metric serves as a tool to track the progress of graph construction. For a fixed iteration count $t$, a larger value of $M$ indicates faster pruning and more rapid expansion of the processed set. For a fixed $M$, a smaller value of $t$ implies that the algorithm has processed $M$ points more quickly. 


\section{Degree Bounds of SNG}
\label{sec:degree}
In this section, we establish upper bounds on the degree of SNG. As shown in Lemma~\ref{lemma:prop}, the out-degree of each point equals the number of pruning iterations during construction, since exactly one edge is added in each iteration. Consequently, bounding the degree reduces to analyzing the dynamics of the pruning process and estimating the total number of iterations required for graph construction. 

\rev{
\textbf{A preview of the main result.} Before delving into the technical analysis, we state the central conclusion of this section: with probability $1$, the maximum out-degree of the constructed SNG is $O(n^{2/3+\varepsilon})$ for any small $\varepsilon>0$ (Theorem~\ref{theorem:max_degree}, Section~\ref{sec:main-degree}). The remainder of Section~\ref{sec:degree} formalizes the two-phase pruning dynamics needed for this bound: Lemma~\ref{lem:probability} (pruning probability), Lemma~\ref{lem:2-dim} (2D base case), Lemma~\ref{lem:sublinear} (sublinear progress in $d$ dimensions), and finally Theorem~\ref{theorem:max_degree}.
}\revnote{R3.D3} 

For completeness, all detailed proofs of lemmas and theorems are provided in the full version of this paper on our GitHub repository~\cite{github} and arxiv~\cite{Ma2025}.
\subsection{Two-Phase Behavior: Empirical Motivation}
\label{subsec:2-phase}
To motivate our theoretical analysis, we first present empirical evidence revealing a distinctive \emph{two-phase behavior} in SNG construction. Figure~\ref{fig:two-case} illustrates two representative cases from the SIFT1M dataset: a fast case completing in only 31 iterations and a slow case requiring over 1000 iterations. Despite this disparity in total iterations, both cases exhibit a consistent structural pattern: an \underline{initial phase of rapid candidate reduction}, followed by a prolonged \underline{plateau phase} during which progress slows significantly. 

This empirical observation strongly suggests that a phase-wise analysis is both natural and necessary for deriving degree bounds. In the remainder of this section, we formalize this intuition. Specifically, Lemma~\ref{lem:sublinear} establishes rapid progress in the initial phase, which in turn enables a degree bound in Theorem~\ref{theorem:max_degree}.

\begin{figure}[t]
    \centering
    \begin{subfigure}[b]{0.48\linewidth}
        \centering
        \includegraphics[width=0.98\linewidth]{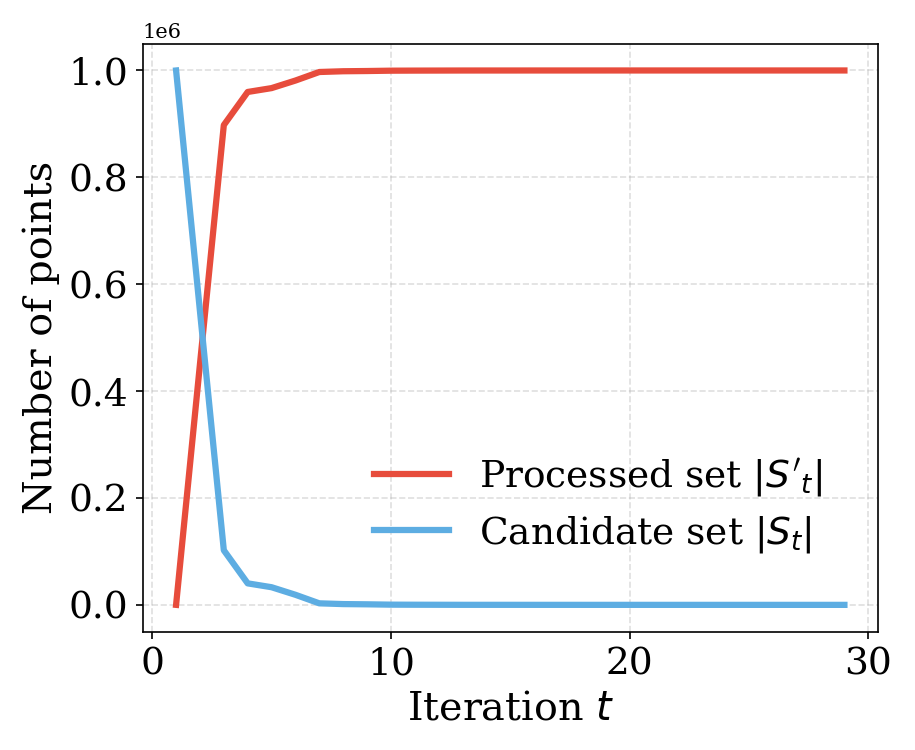}
        \caption{Fast pruning}
        \label{fig:fast_pruning}
    \end{subfigure}%
    \hfill  
    \begin{subfigure}[b]{0.48\linewidth}
        \centering
        \includegraphics[width=0.98\linewidth]{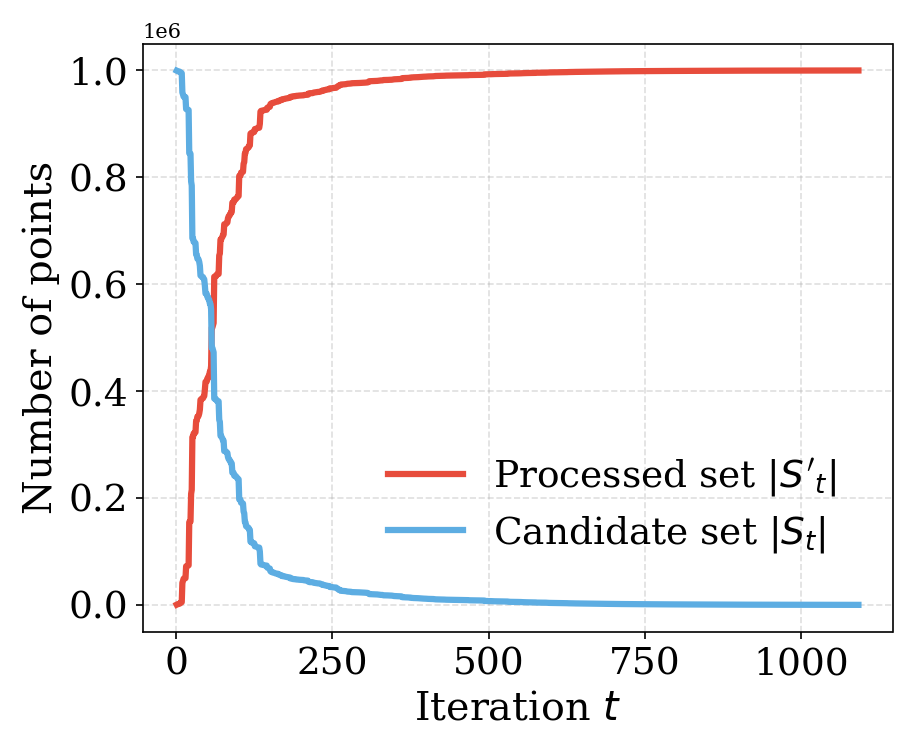}
        \caption{Slow pruning}
        \label{fig:slow_pruning}
    \end{subfigure}
    \caption{Evolution of processed and candidate sets during SNG construction. (a) shows a fast pruning case terminating in 31 iterations. (b) shows a slow pruning case requiring over 1,000 iterations.}
    \label{fig:two-case}
\end{figure}
\vspace{-2pt} 
\subsection{Pruning Probability Analysis}
We begin by analyzing the pruning probability at each iteration, which governs the rate at which candidate points are eliminated. This is a key ingredient for quantifying the speed of the initial fast pruning phase.
\begin{lemma}[Pruning Probability in SNG Construction]
\label{lem:probability}
Given a dataset \( P \subset \mathbb{R}^d \) consisting of \( n \) points independently and uniformly distributed within a ball of radius \(\rho_0\), centered at \(p\), we consider the SNG construction processed on the center point \(p\) with \( \alpha = 1 \). The probability \(\pi_t\) that a point \(p'\) is pruned in the \(t-\)th iteration is given by
\[
\pi_t = \frac{ I_{1 - (\rho_t/2\rho_0)^2}\left(\frac{d+1}{2}, \frac{1}{2}\right) - (\rho_t/\rho_0)^d \cdot I_{0.75}\left(\frac{d+1}{2}, \frac{1}{2}\right) }{2 \left(1 - (\rho_t/\rho_0)^d\right)},
\]
where \(I_x(a, b)\) denotes the \textit{regularized incomplete Beta function}, and \(\rho_t\) is the distance between the nearest neighbor and \(p\) in \(t-th\) iteration.
For large \(d\), this probability satisfies
\[
\frac{ I_{0.75}\left(\frac{d+1}{2}, \frac{1}{2}\right)}{2} < \pi_t < \frac{1}{2}.
\]
\end{lemma}
\rev{The complete proof, including the volume computation via the regularized incomplete Beta function and the lower bound for general $d$, is provided in Appendix~\ref{app:probability}. in \cite{Ma2025}}\revnote{R1.D1}
\subsection{Two-Dimensional Case}
To build intuition, we first analyze a low-dimensional setting using the $M$-$t$ level. In particular, for the two-dimensional case, we show that \emph{only four iterations} suffice to prune over 80\% of all candidate points with high probability.
\begin{lemma}[Two-Dimensional Case]
\label{lem:2-dim}
Given a dataset \( P \subset \mathbb{R}^2 \) with \(n\) points independently and uniformly distributed within a ball of radius \(\rho_0\). Then, with probability \(1-o(1)\), the SNG construction process reaches \((0.8(n-1))\)-\(4\) level.
\end{lemma}
\vspace{-0.3cm}
\begin{proof}
Recall that in the SNG construction of \(p \in P\), we iteratively prune a candidate set \( S_i \), starting with \( S_0 = P\setminus{\{p\}} \). \rev{%
The pruning event $A_{p'}=\{\|p-p'\|\ge\alpha\|p_t^*-p'\|\}$ depends only on
the relative position of the single candidate $p'$ to the two reference
points $(p,p_t^*)$. Conditional on $p_t^*$, the remaining candidates are
i.i.d.\ samples and their relative positions to $(p,p_t^*)$ do not influence
one another. Hence the events $\{A_{p'}\}$ are mutually independent given
$p_t^*$, and $\Delta S_t  \sim \mathrm{Bin}(|S_{t-1}|-1,\pi_t)$ for given \(p_t^*\).
}
\revnote{R2.D2}
Let \( \Delta S_i \) denote the number of points pruned at iteration \( i \). 
 After \( t \) iterations, the processed set \( S'_t \) consists of all pruned points and the selected nearest neighbors. Therefore,
\(
|S'_t| = \sum_{i=1}^t (\Delta S_i + 1),\ t = 2,3,\ldots,
\)
with \( |S'_1| = \Delta S_1 + 1 \), since each iteration contributes exactly one selected nearest neighbor in addition to the pruned points.

From Lemma~\ref{lem:probability}, we know that for \( d = 2 \), the pruning probability at each iteration satisfies \( \pi_i > 1/3 \). To estimate the pruning progress, we conservatively assume a uniform lower bound \( \pi_{\min} = 1/3 \) across all iterations.

The expected number of pruned points at iteration \( i \) is
\(
\mathbb{E}[\Delta S_i] = \pi_i (|S_{i-1}| - 1),
\)
and the variance is
\(
\mathrm{Var}(\Delta S_i) = \pi_i (1 - \pi_i)(|S_{i-1}| - 1).
\)
Since \( \pi_i (1 - \pi_i) < 1/3 \), the normalized variance can be bounded as
\begin{align}
\mathrm{Var}\left( \frac{\Delta S_i}{n-1} \right)
= \frac{\pi_i (1 - \pi_i)(|S_{i-1}| - 1)}{(n-1)^2}
\le \frac{1}{3(n-1)}.    
\end{align}
Under the minimal pruning rate assumption, the expected fraction of points pruned at each iteration is at least
\[
\pi_{\min}^1 = \tfrac{1}{3}, \quad
\pi_{\min}^2 = \tfrac{2}{9}, \quad
\pi_{\min}^3 = \tfrac{4}{27}, \quad
\pi_{\min}^4 = \tfrac{8}{81}.
\]
Consequently, the cumulative expected fraction of processed points after four iterations exceeds \(0.8\).

To ensure that this behavior holds with high probability, we control the deviation of \( \Delta S_i \) from its expectation by choosing constants \( \epsilon_i > 0 \) such that
\[
\epsilon_1 < \frac{\mathbb{E}[\Delta S_1]}{n-1} - \frac{1}{3}, \quad
\epsilon_2 < \frac{\mathbb{E}[\Delta S_2]}{n-1} - \frac{2}{9},
\]
\[
\epsilon_3 < \frac{\mathbb{E}[\Delta S_3]}{n-1} - \frac{4}{27}, \quad
\epsilon_4 < \frac{\mathbb{E}[\Delta S_4]}{n-1} - \frac{8}{81}.
\]
Applying Chebyshev’s inequality yields, for each iteration \( i \),
\begin{align}
\Pr\left(
\left|
\frac{\Delta S_i}{n-1}
-
 \frac{\pi_i(|S_{i-1}| - 1)}{n-1}
  \right|
  \le \epsilon_i
  \right)
  &\ge
  1 - \frac{\mathrm{Var}\left( \frac{\Delta S_i}{n-1} \right)}{\epsilon_i^2}
  \\&\ge
  1 - \frac{1}{3(n-1)^2 \epsilon_i^2}.    
\end{align}
Assuming that all four deviations remain within their respective bounds, the probability that the cumulative number of processed points satisfies \( |S'_4| \ge 0.8(n-1) \) is bounded below by
\begin{align}
\Pr(|S'_4| \ge 0.8(n-1))
\ge
\prod_{i=1}^4
\left(1 - \frac{1}{3(n-1)^2 \epsilon_i^2}\right)
= 1 - o(1).    
\end{align}
Therefore, with probability \( 1 - o(1) \), the SNG construction reaches the \( (0.8(n-1)) \)-4 level.
\end{proof}

\rev{\subsection{High-Dimensional Settings}\revnote{R1.O1}
\label{subsec:highdim}
We now extend the analysis to general $d$-dimensional space. The same phenomenon as in the 2D case manifests as \emph{sublinear-time progress}: a linear number of points is processed within sublinear iterations. This behavior is formalized in the following lemma.}
\begin{lemma}[Sublinear Time Progress]
\label{lem:sublinear}
Given a dataset \( P \subset \mathbb{R}^d \) with \( n \) points distributed uniformly, for all relaxation parameter \(
\alpha>0
\):\\
(1) For any constants \( \nu_1, \nu_2 \in (0,1) \) with \( \nu_1 > \nu_2 \), with probability one, the SNG construction reaches the \( n^{\nu_1} \)-\( O(n^{\nu_2}) \) level.\\
(2) For any constant \( \nu \in (\frac{2}{3},1) \), with probability one, the SNG construction reaches the \( (n - n^{1-\nu}) \)-\( O(n^{\nu}) \) level.
\end{lemma}

\begin{proof}
Part (1).
Fix constants \(\nu_1>\nu_2\) and define the first-passage time
\(
t_1 := \inf\{t\ge 0:|S'_t|\ge n^{\nu_1}\}.
\)
We prove that $t_1=O(n^{\nu_2})$ \rev{almost surely \footnote{Throughout this paper, "almost surely" (abbreviated \emph{a.s.}) is used in its standard measure-theoretic sense: the event holds with probability $1$, i.e., the set of outcomes on which it fails has probability measure zero.}}\revnote{R1.D2}, i.e., there exists a constant
$C>0$ such that
\(
\Pr(t_1\le Cn^{\nu_2})=1.
\)
Equivalently, it suffices to show that for any fixed $C>0$,
\begin{align}
\Pr\!\left(|S'_{\lfloor Cn^{\nu_2}\rfloor}|<n^{\nu_1}\ \text{i.o.}\right)=0,    
\end{align}
where "i.o.'' (infinitely often) denotes the event
\begin{align}
\limsup_{n\to\infty}
\left\{|S'_{\lfloor Cn^{\nu_2}\rfloor}|<n^{\nu_1}\right\}
=
\bigcap_{n_0=1}^{\infty}\ \bigcup_{n\ge n_0}
\left\{|S'_{\lfloor Cn^{\nu_2}\rfloor}|<n^{\nu_1}\right\}.    
\end{align}
Let $t=\lfloor Cn^{\nu_2}\rfloor$. Since
\(
|S'_t|=\sum_{i=1}^t (\Delta S_i+1),
\)
we obtain the deterministic lower bound
\(
|S'_t|
\;\ge\;
t \cdot \min_{1\le i\le t}(\Delta S_i+1),
\)
which implies
\begin{align}
\Pr(|S'_t|<n^{\nu_1})
\;\le\;
\Pr\!\left(
\min_{1\le i\le t}\Delta S_i
<
\frac{n^{\nu_1}-t}{t}
\right).    
\end{align}
Since $\nu_1>\nu_2$, for sufficiently large $n$ we have $n^{\nu_1}>2t$, and hence
\begin{align}
\frac{n^{\nu_1}-t}{t}
\;\ge\;
\frac{1}{2C}\,n^{\nu_1-\nu_2}.    
\end{align}
Recall that $\Delta S_i\sim \mathrm{Bin}(|S_{i-1}|-1,\pi_i)$ with
\(
\mathbb{E}[\Delta S_i]=\pi_i(|S_{i-1}|-1).
\)
Moreover, the sampling probability admits the uniform lower bound
\(
\pi_i \ge \tfrac12 I_{0.75}\!\left(\tfrac{d+1}{2},\tfrac12\right).
\)
Applying the Chernoff bound~\cite{Chernoff}, for any $\delta\in(0,1)$,
\begin{align}
\Pr\!\big(\Delta S_i\le (1-\delta)\mathbb{E}[\Delta S_i]\big)
\;&\le\;
\exp\!\Big(-\tfrac{\delta^2}{2}\mathbb{E}[\Delta S_i]\Big)\\
\;&\le\;
\exp\!\Big(
-\tfrac{\delta^2}{4}\big(|S_{i-1}|-1\big)\,
I_{0.75}\!\left(\tfrac{d+1}{2},\tfrac12\right)
\Big).    
\end{align}
For all $i\le t$ such that $|S_{i-1}|\ge n/2$, we have
$\mathbb{E}[\Delta S_i]\ge c_0 n$ for some constant $c_0>0$.
Choosing $\delta\in(0,1)$ such that
\begin{align}
(1-\delta)\mathbb{E}[\Delta S_i]
\;\ge\;
\frac{1}{2C}\,n^{\nu_1-\nu_2},    
\end{align}
we obtain
\(
\Pr(|S'_t|<n^{\nu_1})
\;\le\;
\exp(-c n)
\)
for some constant $c>0$ and all sufficiently large $n$.
Consequently,
\begin{align}
\sum_{n=1}^{\infty}
\Pr\!\left(|S'_{\lfloor Cn^{\nu_2}\rfloor}|<n^{\nu_1}\right)
<\infty.    
\end{align}
By the Borel--Cantelli lemma (see Lemma~\ref{lem:BC} in Appendix~\ref{app:sup}),
\begin{align}
\Pr\!\left(|S'_{\lfloor Cn^{\nu_2}\rfloor}|<n^{\nu_1}\ \text{i.o.}\right)=0,    
\end{align}
which completes the proof of Part~(i).

Part (2).
Let $\nu\in(2/3,1)$ and define
\(
t_2:=\inf\{t\ge 0:\ |S'_t|\ge n-n^{1-\nu}\}.
\)
We show that $t_2=O(n^\nu)$ almost surely.
Fix $C>0$ and set $t=\lfloor Cn^\nu\rfloor$. As above,
\begin{align}
\Pr\!\left(|S'_t|<n-n^{1-\nu}\right)
\;\le\;
\Pr\!\left(
\min_{1\le i\le t}\Delta S_i
<
\frac{n-n^{1-\nu}-t}{t}
\right).    
\end{align}
For sufficiently large $n$, we have $t\le n/2$ and thus
\begin{align}
\frac{n-n^{1-\nu}-t}{t}
\;\ge\;
\frac{1}{2C}\,n^{1-\nu}.    
\end{align}
Using the same Chernoff bound and the same uniform lower bound on $\pi_i$,
we again obtain an exponentially small upper bound
\begin{align}
\Pr\!\left(|S'_t|<n-n^{1-\nu}\right)
\;\le\;
\exp(-c' n)    
\end{align}
for some constant $c'>0$.
Therefore,
\begin{align}
\sum_{n=1}^{\infty}
\Pr\!\left(|S'_{\lfloor Cn^\nu\rfloor}|<n-n^{1-\nu}\right)
<\infty.    
\end{align}
Applying the Borel--Cantelli lemma yields
\begin{align}
\Pr\!\left(|S'_{\lfloor Cn^\nu\rfloor}|<n-n^{1-\nu}\ \text{i.o.}\right)=0,    
\end{align}
and hence $\Pr(t_2=O(n^\nu))=1$.

\end{proof}
\vspace{-1pt} 
\subsection{Main Degree Bound}
\label{sec:main-degree}
Lemma~\ref{lem:sublinear} establishes that the initial phase rapidly processes a large fraction of points. However, prior empirical observations suggest that this rapid progress does not persist throughout the entire construction. The pruning process undergoes a \textit{significant slowdown} after the initial fast pruning phase, often reaching a \textit{plateau phase}. 

To analyze this plateau phase, we exploit a key fact: the submartingale \({|S'_t|}_{t\ge 0}\) has bounded per-step variance along the pruning trajectory. This enables us to apply the Differential Equation Method (DEM)~\cite{DEM1,DEM2,DEM3}, a standard technique for approximating bounded-variance stochastic processes by deterministic trajectories. DEM yields an ordinary differential equation that captures the expected evolution of \(|S'_t|\) in the plateau regime. Combining the initial fast-pruning bound from Lemma~\ref{lem:sublinear} with the DEM-based characterization of the plateau phase gives a unified view of the pruning dynamics.
We now sketch the proof of our main result, Theorem~\ref{theorem:max_degree}. \rev{Detailed proofs of all lemmas and theorems are provided in the full version of this work on our GitHub repository~\cite{github} and arxiv~\cite{Ma2025}.}
\begin{theorem}
\label{theorem:max_degree}
Given a dataset $P$ of $n$ points in $d$-dimensional space, with probability $1$, the maximum out-degree in the constructed SNG is bounded by $O(n^{2/3 + \epsilon})$ for any small constant $\epsilon > 0$.
\end{theorem}
\begin{proof}[Proof sketch]
\color{blue}{
The proof consists of two parts: (i) establishing a uniform lower bound
\(\pi_t\ge c_*>0\) on the one-step pruning probability, and (ii) applying the
two-phase pruning analysis from Sections~\ref{subsec:2-phase}--\ref{subsec:highdim},
which combines Chernoff bounds, Borel--Cantelli, and Wormald's differential
equation method (DEM). Since the second part only uses the existence of such a
positive lower bound, the same argument applies once \(c_*\) is verified under
each distributional setting.
\revnote{R1.O2\\R2.D1\\R3.D2}

\medskip
\noindent\textbf{Case~1: uniform distribution and \(\alpha=1\).}
For \(\alpha=1\), the blocking condition
\(
\|p-x\|\ge \|p_t^*-x\|
\)
defines a spherical-cap region. By Lemma~\ref{lem:probability}, its normalized
volume is bounded below by a positive constant:
\(
\pi_t\ge c(d,1)>0.
\)
Let \(Y_t:=|S'_t|\). Since \(|S_t|=n-1-Y_t\),
\(
Y_{t+1}-Y_t=1+\Delta S_{t+1},
\)
and hence, conditionally on \(\mathcal{F}_t\),
\(
\mathbb{E}[Y_{t+1}-Y_t\mid\mathcal{F}_t]
\ge 1+c(d,1)(n-1-Y_t).
\)
For \(Z_t:=Y_t/n\), the normalized drift is therefore governed, up to lower-order
terms, by
\(
F(z)=1+c(d,1)(1-z).
\)
The DEM conditions hold: boundedness follows from
\(
|Y_{t+1}-Y_t|\le |S_t|\le n,    
\)
the trend condition follows from the above drift estimate, and \(F\) is Lipschitz since
\begin{align}
|F(z_1)-F(z_2)|=c(d,1)|z_1-z_2|.    
\end{align}
Thus, Wormald's DEM applies.

By Lemma~\ref{lem:sublinear}, for any \(\nu\in(2/3,1)\), the first phase reaches
\(
|S'_t|\ge n-n^{1-\nu}
\)
within \(t_1=O(n^\nu)\) iterations with probability one. At that point at most
\(n^{1-\nu}\) candidates remain. Applying the DEM-controlled plateau analysis
to the residual process gives \(O(n^{1-\nu})\) additional iterations. Therefore,
\begin{align}
T=O(n^\nu)+O(n^{1-\nu})=O(n^\nu).   
\end{align}
Taking \(\nu=2/3+\varepsilon\) yields
\(
T=O(n^{2/3+\varepsilon}).
\)

\medskip
\noindent\textbf{Case~2: uniform distribution and general \(\alpha>0\).}
For general \(\alpha\), the blocking condition becomes
\(
\|p-x\|\ge \alpha\|p_t^*-x\|,
\)
which defines an Apollonius-type blocking region. Its normalized volume is still
bounded below by a positive constant \(c(d,\alpha)>0\), so
\(
\pi_t\ge c(d,\alpha).
\)
The proof above uses the geometry only through this drift lower bound. Therefore
the same argument applies with \(c(d,1)\) replaced by \(c(d,\alpha)\), yielding
\(
T=O(n^{2/3+\varepsilon})
\)
for every fixed \(\alpha>0\).

\medskip
\noindent\textbf{Case~3: bounded non-uniform density.}
Suppose the local candidate density \(f\) satisfies \(0<m\le f\le M<\infty\).
Conditional on \(p_t^*\), the blocking region \(R_t(p_t^*)\) from Case~2 gives
\[
\pi_t
\ge
\frac{m}{M}
\frac{\operatorname{Vol}(R_t(p_t^*))}
     {\operatorname{Vol}(\Omega_{t-1})}
\ge
\frac{m}{M}c(d,\alpha)
=:c_*>0.
\]
Thus non-uniformity only changes the constant in the drift lower bound. The same
Chernoff, Borel--Cantelli, and DEM argument then gives
\(
T=O(n^{2/3+\varepsilon}).
\)

Since the out-degree is controlled by the number of pruning iterations, the
claimed degree bound follows. Detailed calculations for the \(\alpha\)-extension
and the bounded-density case are provided in Appendix~B.3 in \cite{Ma2025}.
}
\end{proof}
\section{Search Bounds}
\label{sec:path}
We now turn to the analysis of query processing on the constructed SNG. Our goal is to show that, with high probability, \textit{GreedySearch} converges to an approximate nearest neighbor in \(O(\log n)\) steps. 
\begin{theorem}
\label{theorem:path}
\rev{Given a dataset \( P \subset \mathbb{R}^d \) with \( n \) points distributed uniformly,} with probability $1$, the SNG search converges in $O(\log n)$ steps.
\end{theorem}\revnote{R1.O3}
\vspace{-0.2cm}
\begin{proof}[Sketch Proof]
Let $\{v_0, v_1, \dots, v_k\}$ denote a greedy search path from the start node $p = v_0$ to the approximate nearest neighbor of the query $q$. At each step, \textit{GreedySearch} moves to the neighbor closest to $q$, ensuring that $\|v_i - q\| > \|v_{i+1} - q\|$.

To analyze the expected path length $k$, we analyze the distribution of neighbors around the destination \(v_k\), they are uniformly distributed in a $d$-dimensional ball of radius \(\rho_0\). We partition this ball into concentric annular layers and let $\eta_i$ denote the number of neighbors lying in the $i$-th layer. Each $\eta_i$ follows a binomial distribution proportional to the shell volume, and the total number of neighbors is bounded by $\eta = O(n^{2/3+\epsilon})$ from Theorem~\ref{theorem:max_degree}.

Using volume arguments, we upper bound the expected total number of neighbors located within the ball of radius \(\|v_k-v_0\|\) by the expected number of dataset points that fall into the same region. Specifically, we show 
\begin{align}
\mathbb{E}[k \cdot \eta] \leq \frac{n (n+1)^{(d-1)/d} \rho_0^{d+1} \cdot E\left[ \left( (\rho_0 - \triangle r)/\rho_0 \right)^{d(k-1)} \right]}{\triangle r^{d+1} \cdot d},    
\end{align}
where $\triangle r$ is the minimum distinguishable pairwise distance difference among all point triples. Applying Jensen’s inequality~\cite{Chernoff}, we define an auxiliary function $g(x)$ 
\begin{align}
    g(x) := \eta \cdot x - \frac{n(n+1)^{(d-1)/d} \rho_0^{d+1}}{d \cdot \triangle r^{d+1}} \left[ -\left( \frac{\rho_0 - \triangle r}{\rho_0} \right)^{d(x-1)} + 2 \right],
\end{align}
whose root bounds $\mathbb{E}[k]$. We demonstrate $g(k) < 0$, note that \(g(x)\) is monotonically increasing, and $g(\log n) > 0$, which together imply that \(k< \log n\), hence $\mathbb{E}[k] = O(\log n)$.
\end{proof}

\section{Truncation Parameter Optimization}
\label{sec:algorithm_optimization}

In the preceding sections, we establish theoretical guarantees for the SNG structure, including sublinear degree growth and logarithmic search complexity. We now leverage these insights to address a practical issue in index construction: how to efficiently select the truncation parameter $R$, which directly controls the trade-off between construction cost and search performance.

\subsection{Complexity Analysis} 
\label{subsec:complexity}
We summarize the full SNG construction pipeline as implemented in practice, incorporating some widely used optimizations~\cite{NEURIPS2019,NSG19}. The process begins by initializing a randomized $R$-regular graph. For each data point $p$, a \textit{GreedySearch} procedure is performed from a fixed entry point (e.g., the dataset centroid) to collect an initial candidate set for pruning, with time complexity
\(
T_{\text{search}} = O(R \log n).
\)
Importantly, in practical implementations the initial candidate set is \emph{not} the entire dataset; instead, it consists only of the points returned by \textit{GreedySearch}. This design significantly reduces construction cost compared to the theoretical baseline that assumes full connectivity. In particular, the size of the initial candidate set satisfies $|S_0| = O(R \log n)$ in expectation.
The candidate set is then refined through iterative SNG pruning, where the truncation parameter $R$ bounds the number of pruning iterations. Let $|S_0|$ denote the initial candidate set size; the pruning cost is
\(
T_{\text{prune}} = O(R \cdot |S_0|).
\)
To further enhance graph connectivity and preserve monotonic searchability, reverse edges are inserted after pruning. These insertions may temporarily violate the degree constraint, triggering additional pruning steps. The amortized cost of reverse-edge handling can be modeled as
\(
T_{\text{reverse}} = O(\alpha \cdot R^3),
\)
where $\alpha$ is the pruning parameter controlling sparsification aggressiveness.

The computational cost of each stage—\textit{GreedySearch}, SNG pruning, and reverse-edge handling—can be analyzed independently. Aggregating these components yields the overall construction complexity, as expressed below.
\[
n\bigl(
C_1 \cdot R \log n + b_1
+ C_2 \cdot (R^2 \log n / \alpha) + b_2
+ C_3 \cdot (\alpha R^3) + b_3
\bigr),
\]
where $C_i$ and $b_i$ are implementation-dependent constants.

This expression highlights the competing effects of $R$ and $\alpha$: increasing $R$ improves graph quality but amplifies both pruning and reverse-edge costs.
 
\subsection{Optimization Principle}
\label{subsec:optimization}
As discussed in Section~\ref{sec:2.3}, $\alpha$ is typically fixed in the range $1\le\alpha\le2$, whereas the truncation parameter $R$ controls the trade-off between construction cost and search robustness: larger $R$ usually improves recall but increases the cost.

\rev{We determine $R$ for a target $\alpha$ using a \textit{marginal optimality principle}. From Section~\ref{subsec:complexity}, the construction cost can be summarized as
\(
n\bigl(C_1R\log n+C_2R^2\log n/\alpha+C_3\alpha R^3+O(1)\bigr),
\)
where the three terms correspond to initial candidate collection, SNG pruning, and reverse-edge handling, respectively. Treating $R$ as a function of $\alpha$, the dominant stationary condition is
\[
-C_2R^2\log n/\alpha^2+C_3R^3=0,
\]
which gives
\(
R^*\propto \log n/{\alpha^2}.
\)
Thus, the optimal truncation parameter grows logarithmically with $n$ and decreases quadratically with $\alpha$; the hidden constant absorbs the dimension, implementation constants, and lower-order terms. Algorithm~\ref{alg:determine_R} calibrates this constant by a single reference build at $R=n^{2/3}$ under $\alpha_1$, measures the average out-degree $\bar R$, sets $K'=\alpha_1^2\bar R/\log n$, and returns $R^*=K'\log n/\alpha_2^2$ for the target $\alpha_2$. This avoids costly parameter sweeping while using the theoretical degree bound in Theorem~\ref{theorem:max_degree} as a near-untruncated reference.}\revnote{R3.D3}

\begin{algorithm}[t]
\small
\SetVline
\SetKwInput{Input}{Input}
\SetKwInput{Output}{Output}
\caption{Optimization of Truncation Parameter $R$ of SNG}
\label{alg:determine_R}

\Input{test pruning parameter $\alpha_1$, target pruning parameter $\alpha_2$, number of data points $n$}
\Output{optimized truncation parameter $R^*$}

$R \leftarrow n^{2/3}$\tcp*{theoretical upper bound}

Build graph index using \texttt{SNG-Prune}$(\alpha_1, R)$\;

Compute the average out-degree $\bar{R}$ of the constructed graph\;

$K' \leftarrow \dfrac{\alpha_1^2 \cdot \bar{R}}{\log n}$\;

$R^* \leftarrow \dfrac{K' \cdot \log n}{\alpha_2^2}$\;

\Return{$R^*$}
\end{algorithm}

\section{Experiments}
\label{sec:experiments}
We compare a conventional parameter-sweeping pipeline with our \emph{analytical} parameter determination strategy, and assess both construction efficiency and query performance across different graph-based ANNS algorithms and representative datasets.
\subsection{Experimental Setup}
\subsubsection{Datasets}
We evaluate our method on six datasets with diverse dimensionalities and application domains, as summarized in Table~\ref{tab:datasets}. 

The first five datasets are standard benchmarks in the ANNS literature and cover a broad range of real-world applications, including image, audio, and text. \textbf{MSong}~\cite{msong} consists of audio feature vectors for music similarity search. \textbf{DEEP1M}~\cite{DEEP1M} and \textbf{SIFT1M}~\cite{DBLP:journals/pami/JegouDS11} are widely used image-retrieval benchmarks based on deep neural features and SIFT descriptors, respectively. \textbf{GIST1M}~\cite{DBLP:journals/pami/JegouDS11} contains high-dimensional GIST descriptors. \textbf{GloVe}~\cite{glove} provides word embeddings for semantic search. 
In addition, we include a synthetic dataset, \textbf{UNIFORM}, where points are generated uniformly at random in a unit hypercube.
\begin{table}
  \caption{Summary of Datasets Used in Experiments}
  \label{tab:datasets}
  \centering
  \begin{tabular}{l|cccc}
    \toprule
    Dataset & Cardinality & Dimension & Query Size & Data Type \\
    \midrule
    MSong     & 992,272   & 420  & 200  & Audio \\
    DEEP1M    & 1,000,000 & 256  & 1,000 &Image \\
    SIFT1M    & 1,000,000 & 128  & 10,000 & Image \\
    GIST1M    & 1,000,000 & 960  & 1,000 & Image \\
    GloVe     & 1,193,514 & 200  & 1,000 & Text \\
    UNIFORM   & 1,000,000 & 128  & 10,000 & Synthetic \\
    \bottomrule
  \end{tabular}
  \vspace{-0.1cm}
\end{table}
\subsubsection{Performance Metrics}
We measure both index-construction efficiency and query behavior using several indicators. \textbf{Index Construction Time} measures the end-to-end time to obtain a finalized graph index, which includes both the time spent on selecting construction parameters and the subsequent construction phase.

  For query-time evaluation, \textbf{Recall@10} measures the fraction of queries for which at least one true nearest neighbor appears in the top-10 retrieved results, formally defined as $\text{Recall@}k = |A \cap A'|/|A'|$, where $A$ is the retrieved set and $A'$ is the ground truth with $|A|=|A'|=k$. \textbf{Queries Per Second (QPS)} measures throughput as the number of queries processed per second and \textbf{Query Latency} reports the average time per query in milliseconds at different recall levels.

\subsubsection{Parameter Settings}
 For each fixed $\alpha$, we vary the search width $L$ (and analogous search-time controls) to generate recall-latency trade-off. For baseline methods that tune the truncation parameter $R$ via parameter sweeping, we employ a binary-search (or grid-search) strategy guided by a normalized 1:1 weighted loss over query recall and latency, which reflects standard practice in iterative tuning pipelines.  In contrast, our approach determines $R$ directly from the analytical relationship derived in Section~\ref{subsec:optimization}.

\subsubsection{Computing Environment}
All experiments are conducted on a workstation equipped with an Intel CPU, 64GB DDR4-3200 memory, and a 2TB NVMe SSD, running Ubuntu 22.04 LTS. We compile all code with GCC 11.4 and enable \texttt{-O3} optimizations.

\subsubsection{Compared Algorithms}
We evaluate three representative SNG-based ANNS algorithms and apply our optimization principle consistently across them. \textbf{Vamana}~\cite{NEURIPS2019}, which underlies DiskANN-style indexing, relies on SNG-like pruning and exposes the parameters $\alpha$ and $R$ that directly control graph sparsification and maximum out-degree. \textbf{NSG}~\cite{NSG19} constructs a navigable graph with a bounded out-degree constraint and typically requires careful tuning of $R$ to balance construction cost and search quality. \textbf{HNSW}~\cite{HNSW2020} builds a hierarchical multi-layer graph, where the construction-time candidate list size (e.g., \texttt{efConstruction}) plays a similar role in controlling construction overhead and the resulting graph connectivity. We denote the variants using our proposed parameter optimization strategy as \textbf{OPT-DiskANN}, \textbf{OPT-NSG}, and \textbf{OPT-HNSW}, and refer to their original implementations as \textbf{DiskANN}, \textbf{NSG}, and \textbf{HNSW}, respectively. 
\rev{\subsubsection{Additional baselines} To situate our method within the broader competitive landscape, we additionally compare against four recent state-of-the-art ANNS systems that span graph, hashing, and quantization paradigms: \textbf{VSAG}~\cite{DBLP:journals/pvldb/ZhongLJYCWSJGXLSSC25}, a versatile graph-based ANNS library supporting multiple index structures; \textbf{SymphonyQG}~\cite{SymphonyQG}, a recent graph-quantization method achieving high throughput on standard benchmarks; \textbf{DET-LSH}~\cite{hash2}, a locality-sensitive hashing method with dynamic encoding trees; and \textbf{RaBitQ}~\cite{RabitQ}, a quantization-based method with theoretical error bounds. These methods are not modified by OPT-SNG and they serve as external reference points for query performance.}\revnote{R3.D1}

\subsection{Index Construction Performance}
\begin{table}[!t]
  \caption{Index Construction Speedup of OPT-SNG}
  \label{tab:construction_speedup}
  \centering
  \small
  \begin{tabular}{l|ccc|c}
    \toprule
    \multirow{2}{*}{Dataset} & \multicolumn{3}{c|}{Speedup over Baseline} & \multirow{2}{*}{Avg} \\
    & Vamana & NSG & HNSW & \\
    \midrule
    SIFT1M   & 7.9$\times$ & 4.8$\times$ & 3.4$\times$ & 5.4$\times$ \\
    GIST1M   & 5.6$\times$ & 7.1$\times$ & 4.3$\times$ & 5.7$\times$ \\
    DEEP1M   & 6.4$\times$ & 6.4$\times$ & 4.1$\times$ & 5.6$\times$ \\
    MSong    & 4.2$\times$ & 5.3$\times$ & 2.6$\times$ & 4.0$\times$ \\
    GloVe    & 8.7$\times$ & 15.4$\times$ & 2.2$\times$ & 8.8$\times$ \\
    \midrule
    Average  & 6.6$\times$ & 7.8$\times$ & 3.3$\times$ & 5.9$\times$ \\
    \bottomrule
  \end{tabular}
    \vspace{-0.3cm}
\end{table}

We begin by comparing the end-to-end index construction time, including both parameter tuning and the subsequent graph building. Figure~\ref{fig:construction_time_alpha1} reports results across all datasets, where the stacked bars separate tuning time from construction time. Across the board, the most salient difference is the tuning component: parameter sweeping repeatedly triggers expensive partial or full constructions to probe candidate $R$ values, whereas our method largely removes this iterative overhead by computing $R$ analytically once the necessary reference statistics are obtained. As a result, the total construction time decreases substantially even when the actual build phase remains comparable.

Table~\ref{tab:construction_speedup} summarizes the speedups achieved over binary-search baselines. We observe consistent improvements across datasets and algorithms, ranging from $2.2\times$ to $15.4\times$, with an average speedup of $5.9\times$. The gains are particularly pronounced on \textbf{GloVe} , where NSG achieves a $15.4\times$ speedup, and on \textbf{SIFT1M}, where Vamana achieves $7.9\times$. These large gains occur because the binary-search pipeline typically requires more iterations to converge to a good truncation level on these datasets, leading to large cumulative tuning costs.
\begin{table}[t]
    \caption{Recall@10 comparison at fixed latency targets. Latency thresholds are dataset- and algorithm-specific, chosen to reflect typical operating points.}
  \label{tab:recall_comparison}
  \centering
  \small
  \begin{tabular}{l|cc|cc|cc}
    \toprule
    \multirow{2}{*}{Dataset} & \multicolumn{2}{c|}{Vamana} & \multicolumn{2}{c|}{NSG} & \multicolumn{2}{c}{HNSW} \\
    & Binary & OPT & Binary & OPT & Binary & OPT \\
    \midrule
    MSong    & 99.26& \textbf{99.32}& 99.79& \textbf{99.82}& 86.27& \textbf{91.85}\\
    DEEP1M   & 96.77& \textbf{98.76}& 97.94& \textbf{99.33}& 98.93& \textbf{99.71}\\
    SIFT1M   & 99.84& \textbf{99.90}& 99.64& \textbf{99.70}& 99.92& \textbf{99.99}\\
    GIST1M   & 99.63& \textbf{99.76}& 96.16& \textbf{96.79}& 98.20& \textbf{99.63}\\
    GloVe    & 88.91& \textbf{89.94}& 90.34& \textbf{91.83}& 92.97& \textbf{94.61}\\
    UNIFORM  & 98.78& \textbf{99.62}& 93.25& \textbf{95.96}& 96.53& \textbf{97.64}\\
    \bottomrule
  \end{tabular}
\vspace{-0.15cm}
\end{table}
\begin{figure}[!t]
    \centering
    \begin{subfigure}[b]{0.48\linewidth}
        \centering
        \includegraphics[width=0.98\linewidth]{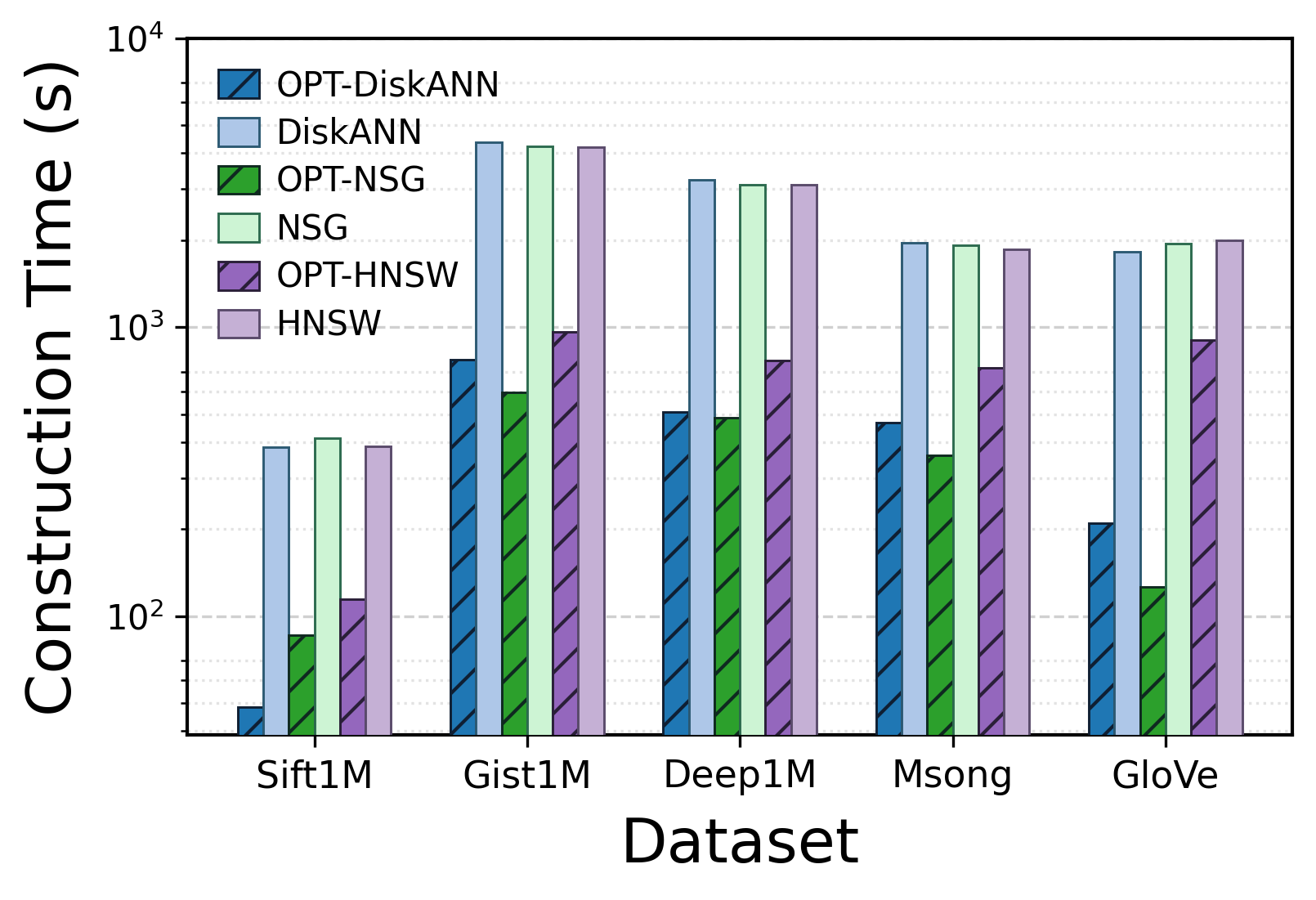}
        \caption{Index construction time comparison}
        \label{fig:construction_time_alpha1}
    \end{subfigure}%
    \hfill
    \begin{subfigure}[b]{0.48\linewidth}
        \centering
        \includegraphics[width=0.98\linewidth]{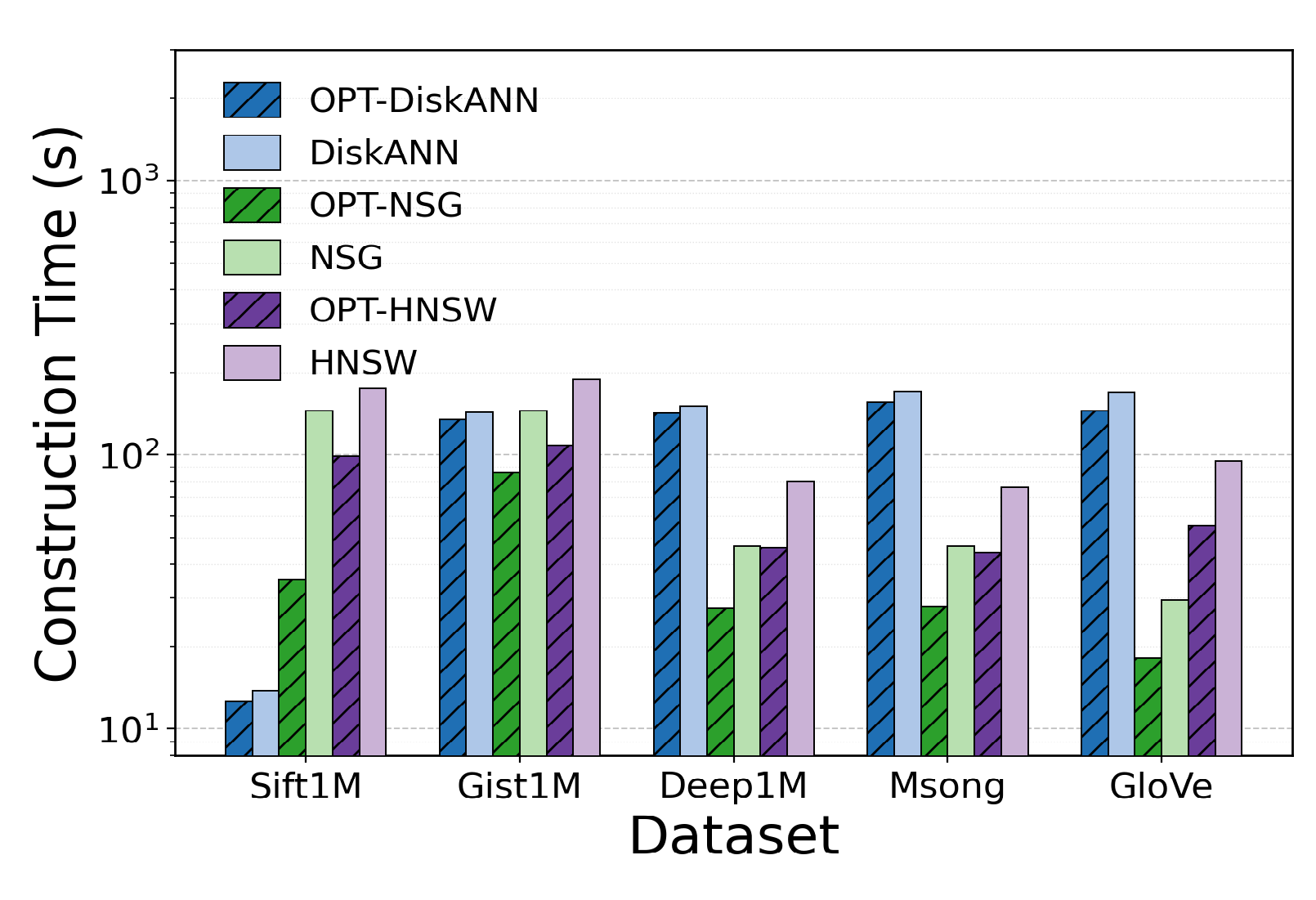}
        \caption{Ablation: Parameter reuse across datasets}
        \label{fig:ablation}
    \end{subfigure}
    \caption{Index construction time comparison and parameter-reuse ablation.}
    \label{fig:construction-ablation}
\end{figure}
\vspace{-0.2cm}

\subsection{Query Performance Evaluation}

We next evaluate whether the analytically chosen parameters (\texttt{OPT}) preserve query-time behavior compared to the parameters selected by parameter sweeping. We report recall and QPS trends by varying the search width $L$ under each $\alpha$.

Figures~\ref{fig:recall-qps} present recall versus QPS curves under different $\alpha$ settings. Across datasets, the curves produced by \texttt{OPT} closely track and in several cases slightly improve upon the curves of the original baseline, indicating that removing iterative tuning does not come at the expense of query quality. Intuitively, this behavior is consistent with our derivation: the analytical rule targets the better construction-quality regime that the parameter sweep attempts to locate empirically, but avoids repeatedly paying the cost of probing suboptimal intermediate choices.

To provide a compact comparison under fixed latency budgets, Table~\ref{tab:recall_comparison} summarizes the achieved Recall at representative, dataset and algorithm specific latency targets. \revnote{R1.O4}
The results show that \texttt{OPT} consistently matches or improves Recall@10 across all three algorithms and all datasets. The improvements are most pronounced on DEEP1M, MSong, GIST1M, and GloVe. 
\rev{This slight improvement comes from better calibrating the truncation
parameter $R$ with respect to the graph-structure parameter $\alpha$. When
$R\ll R^*$, truncation is too aggressive and removes useful SNG edges, reducing
connectivity and recall. When $R\gg R^*$, maintaining $R$-NN candidates and
reverse edges mainly introduces redundant edges and extra overhead, slowing
search without further recall gains. The analytically derived $R^*$ can
slightly improve Recall by avoiding locally suboptimal choices in discrete
parameter sweeps.}
Notably, HNSW exhibits the largest recall gains on SIFT1M, while on the challenging UNIFORM dataset, the most significant improvement is observed for NSG.

Figure~\ref{fig:uniform} shows that the UNIFORM dataset consistently yields lower throughput than SIFT1M across all evaluated algorithms at comparable recall levels. For example, at 99\% Recall@10, Vamana on SIFT1M achieves approximately $7\times$ higher QPS than on UNIFORM. Even under this challenging distribution, \texttt{OPT} consistently achieves higher QPS than the baseline at the same recall levels, demonstrating that our optimization remains effective and preserves query efficiency under uniform data without favorable geometric structure.

\rev{\noindent\textbf{Comparison with Additional Baselines.}\revnote{R3.D1,W1}
Figure~\ref{fig:baseline} also reports four further baselines, covering
the main families of recent ANNS work: two graph-based methods
(VSAG and SymphonyQG), a quantization-based method (RaBitQ), a
hashing-based method (DET-LSH).
Among them, VSAG and SymphonyQG are the most competitive, attaining
high QPS at high recall on SIFT1M and Deep1M; nevertheless, OPT-version
matches or surpasses their throughput across the full Recall--QPS
frontier on every dataset, while requiring substantially less index
construction time (on average \textbf{1.2}$\times$ less than VSAG and
\textbf{1.8}$\times$ less than SymphonyQG). The non-graph baselines
behave very differently: RaBitQ and DET-LSH degrade sharply once
recall enters the high-recall regime and never approaching the graph-based frontier. Taken together, the three OPT variants consistently dominate or match every baseline along the entire
Recall--QPS frontier.}

\subsection{\rev{Ablation: Parameter Reuse Across Builds}}\revnote{R2.D3}
\label{subsec:ablation_reuse}
{\color{blue}
We evaluate a parameter reuse scenario in which the baseline first sweeps $R$ on one dataset and then reuses the resulting value on subsequent builds of the same family, paying no further tuning cost on new data. As shown in Figure~\ref{fig:ablation}, OPT still consistently outperforms the reuse baselines across all five datasets (Sift1M, Gist1M, Deep1M, Msong, GloVe) and three index types (DiskANN, NSG, HNSW). The construction-time reduction averages $1.09\times$ for DiskANN, $1.65\times$ for NSG, and $1.74\times$ for HNSW, and is observed on every (dataset, method) pair. This confirms that the reported construction-time gains stem from improved
parameter quality, and that OPT-SNG still provides direct construction speedups
even after removing the effect of parameter-sweep overhead.
}
\begin{figure*}[t!]
    \centering    
    \includegraphics[width=0.92\linewidth]{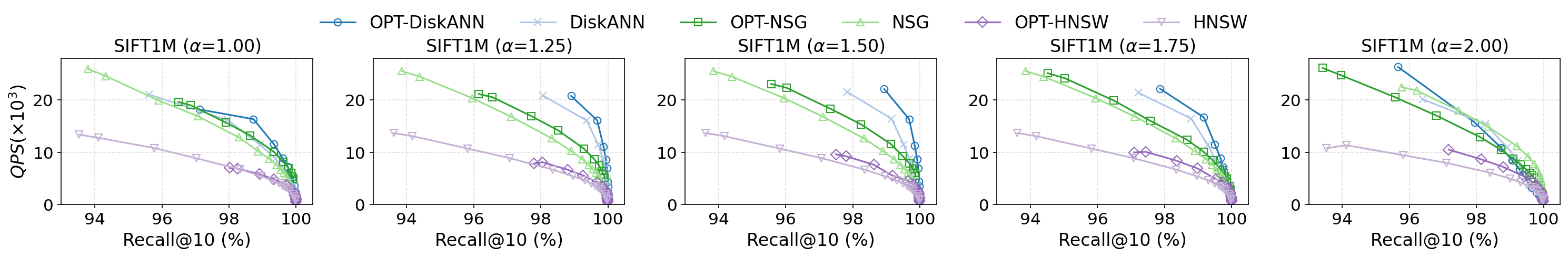}
    \includegraphics[width=0.92\linewidth]{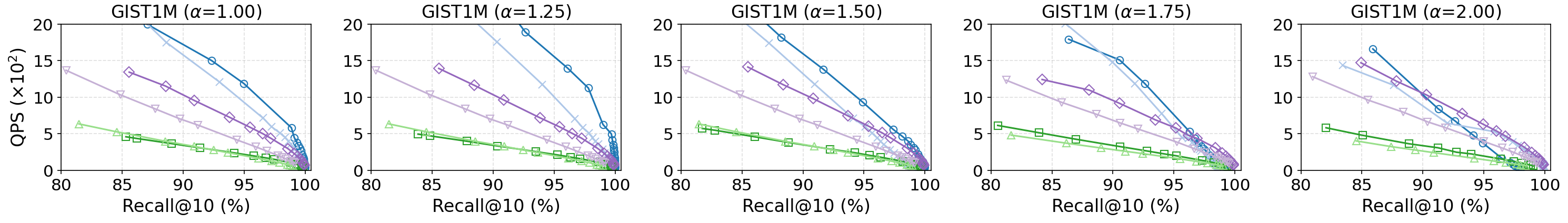}
    \includegraphics[width=0.92\linewidth]{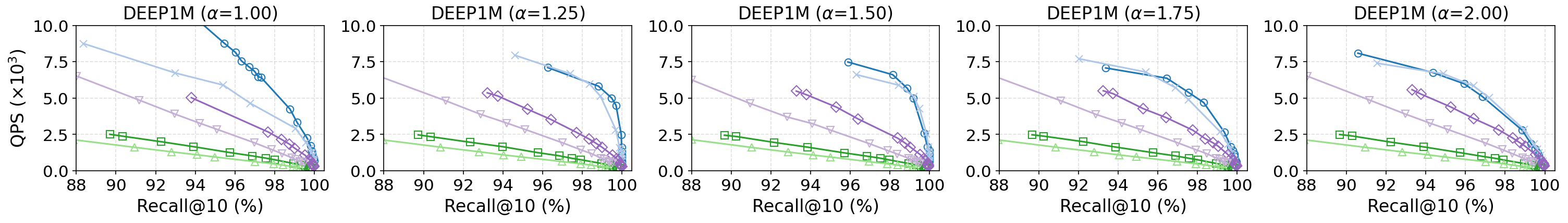}
    \includegraphics[width=0.92\linewidth]{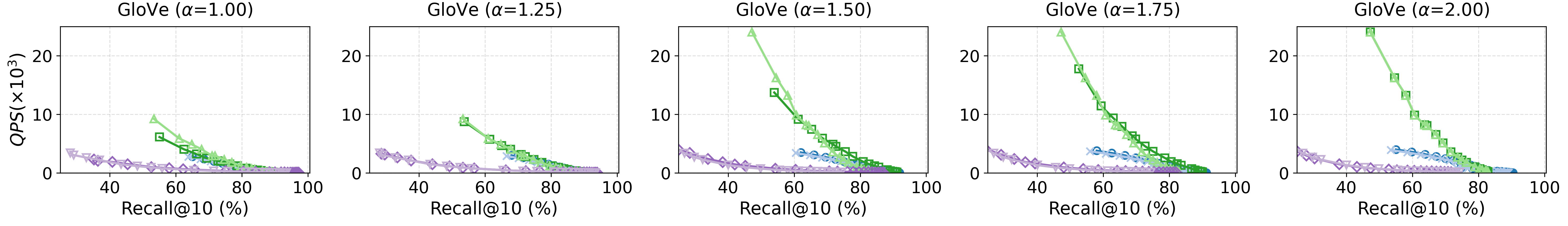}
    \includegraphics[width=0.92\linewidth]{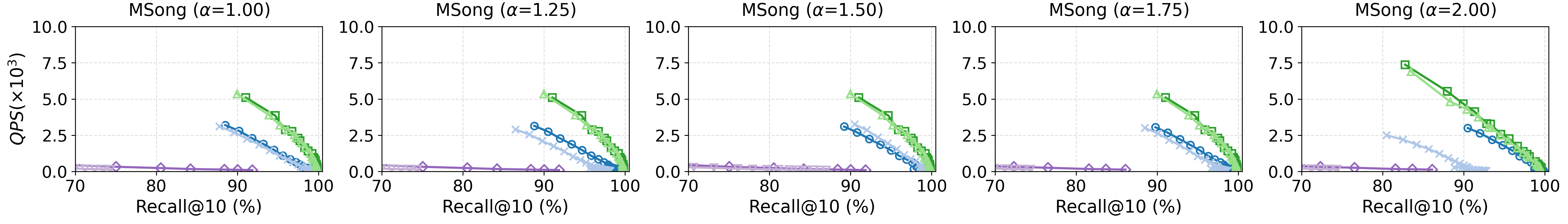}
    \vspace{-0.2cm}
    \caption{Query performance across varying \(\alpha\) values on real-world datasets}
    \label{fig:recall-qps}
    \vspace{-0.2cm}
\end{figure*}

\begin{figure*}[t]
    \centering
    \includegraphics[width=0.92\linewidth]{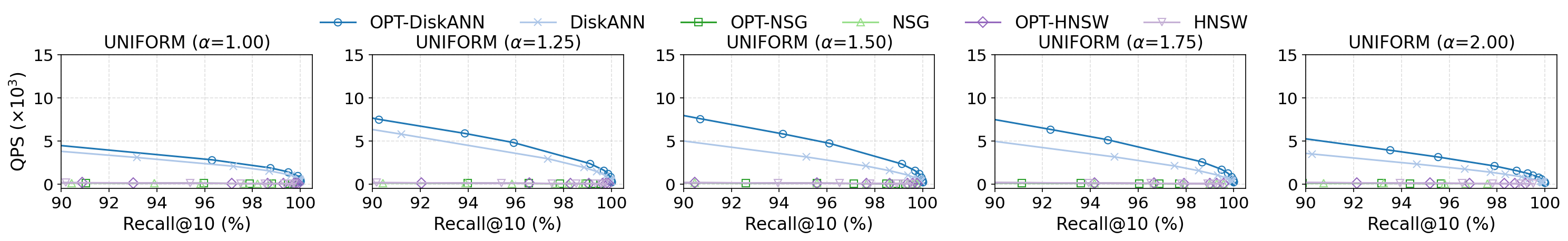}
    \vspace{-0.2cm}
    \caption{Query performance across varying \(\alpha\) values on UNIFORM}
    \label{fig:uniform}
\end{figure*}
 \vspace{-0.2cm}
\subsection{Effect of Pruning Parameter $\alpha$}
Figure~\ref{fig:recall-qps} illustrates the recall versus QPS curves under different $\alpha$ settings. 
 Across almost all datasets, increasing the pruning parameter $\alpha$ exhibits a consistent trend: for a fixed recall level, the achievable QPS initially increases as $\alpha$ grows, and then either saturates or slightly degrades when $\alpha$ becomes too large. We observe that $\alpha = 1.50$ provides the best overall balance across datasets, achieving the highest or near-highest QPS at comparable recall levels in most cases.

 The sensitivity to $\alpha$ varies across datasets. High-dimensional datasets such as GIST1M and DEEP1M exhibit more pronounced performance degradation at large $\alpha$ values, as overly sparse graphs struggle to preserve navigability in high-dimensional spaces. In contrast, lower- or medium-dimensional datasets such as SIFT1M, GloVe, and MSong remain robust across a wider range of $\alpha$ values, showing only mild performance variation.
 \vspace{-0.1cm}


\subsection{Scalability Study}
Finally, we evaluate the scalability of our approach by varying the dataset size and reporting the index construction time. In addition, we conduct experiments on the SIFT10M dataset by sampling subsets of size 2M, 4M, 6M, and 8M points. As shown in Figure~\ref{fig:scalar}, the construction speedups achieved by the \texttt{OPT} variants are consistently preserved across all sampled subsets, and the benefits become increasingly pronounced as the dataset size grows.

Furthermore, as illustrated in Figure~\ref{fig:sift10msearch}, the recall–QPS curves produced by \texttt{OPT} consistently match or outperform those obtained via parameter sweeping across all three algorithms, without compromising search quality. 
\begin{figure*}[t!]
    \centering    
    \includegraphics[width=0.99\linewidth]{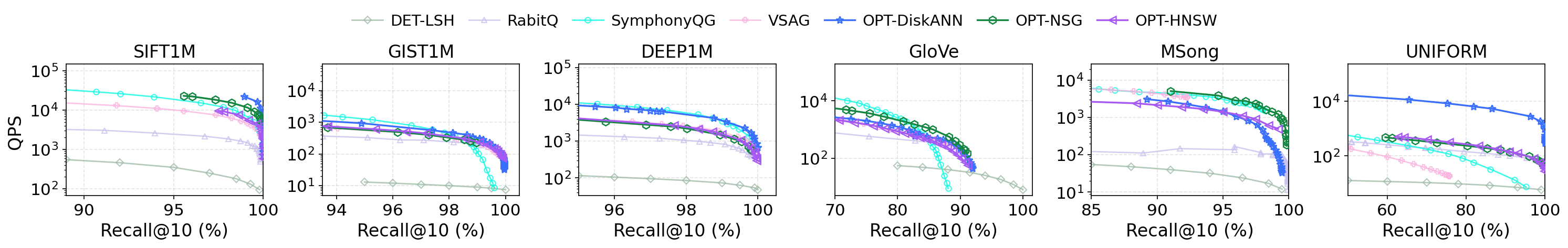}
     \vspace{-0.2cm} 
    \caption{Comparison with additional baselines on six datasets}
    \label{fig:baseline}
\end{figure*}
\begin{figure}[!t]
    \begin{minipage}{0.49\linewidth}
        \centering
        \includegraphics[width=0.96\linewidth]{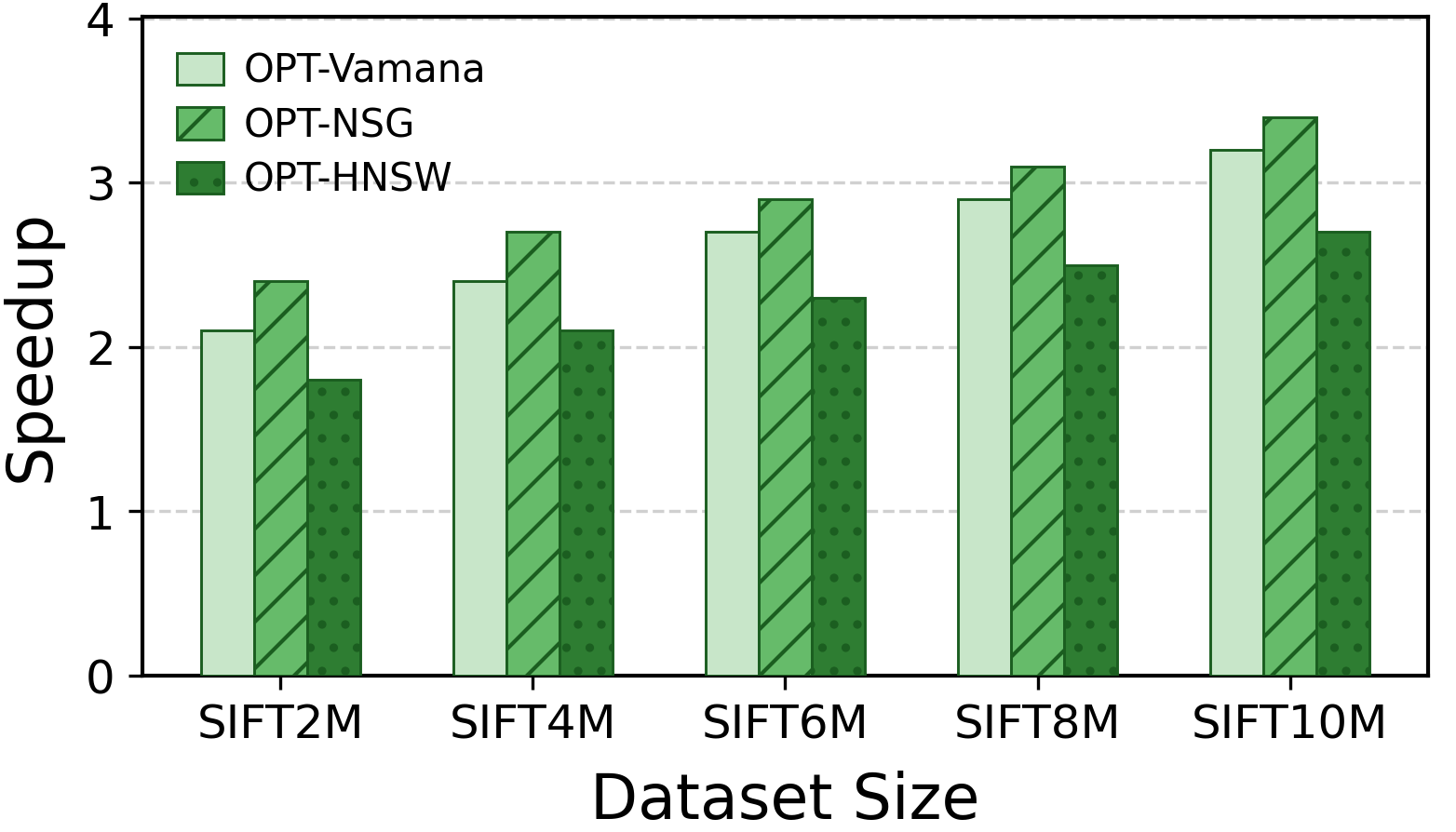}
        \vspace{-0.2cm}
        \caption{Scalability of index construction speedup}
        \label{fig:scalar}
    \end{minipage}
    \hfill
    \begin{minipage}{0.49\linewidth}
        \centering
        \includegraphics[width=0.96\linewidth]{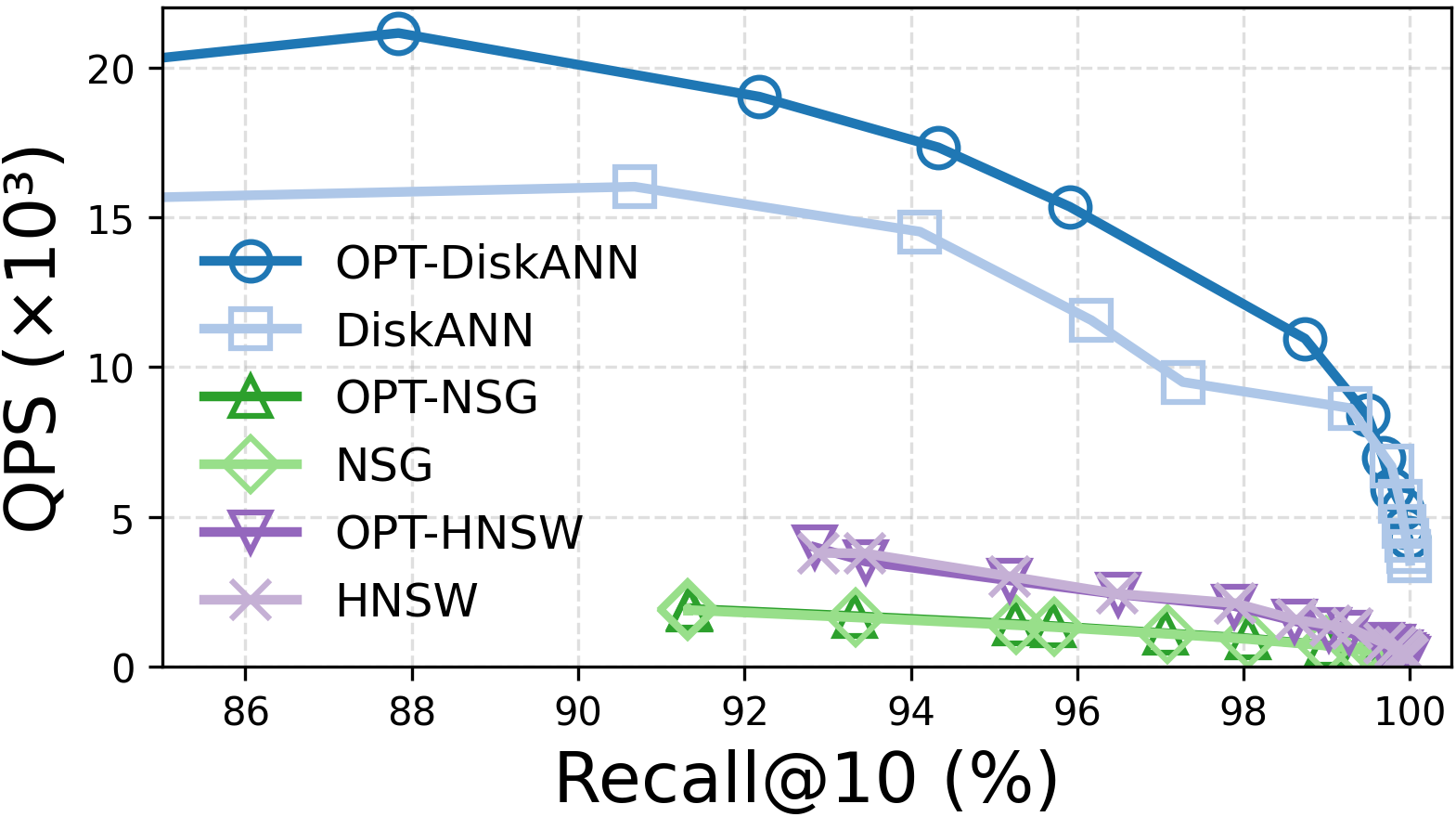}
          \vspace{-0.3cm}
        \caption{Query performance on SIFT10M}
        \label{fig:sift10msearch}
    \end{minipage}  
\end{figure}

{\color{blue}                              
To verify scalability beyond the million-scale setting, we evaluate OPT-Vamana 
on three standard billion-scale benchmarks: \textbf{SIFT1B} (1B vectors, 128-dim), 
\textbf{DEEP1B} (1B vectors, 96-dim), and \textbf{SpaceV1B} (1B vectors, 100-dim), 
following the experimental protocol used by recent large-scale ANNS evaluations 
such as~\cite{tribase}. All baselines and OPT variants are run on the same 
hardware described in Section~\ref{sec:experiments}, and we report results in 
Figure~\ref{fig:billion-sets}.
  
Across all three datasets, OPT-SNG preserves the construction-time advantage observed at the million-point scale. In particular, it achieves graph-construction speedups of \textbf{2.48}$\times$ on SIFT1B, \textbf{1.03}$\times$ on DEEP1B, and \textbf{1.07}$\times$ on SpaceV1B. Moreover, at fixed latency targets, OPT-SNG attains Recall@10 higher than that of the original version on all three datasets.
}                                     

\revnote{R1.O5, R3.W3}  \vspace{-0.2cm}\raggedbottom
\subsection{Summary}
\label{sec:summary}
The experimental results provide evidence supporting the effectiveness of the proposed analytical parameter optimization strategy. Across all evaluated datasets and SNG-based algorithms, our method substantially reduces index construction time by eliminating expensive iterative parameter sweeping. In particular, tuning overhead is reduced by more than an order of magnitude, leading to overall construction speedups of up to $15.4\times$ and an average speedup of $5.9\times$ across datasets.

Crucially, this acceleration does not come at the cost of query quality. The analytically determined parameters consistently preserve—and in many cases improve the recall and QPS compared to baseline parameter sweeping. Recall is matched or improved across all datasets and algorithms at representative latency targets. 
\section{Related Work}
\label{sec:related_work}
\begin{figure}[t]
    \centering
    \includegraphics[width=\linewidth]{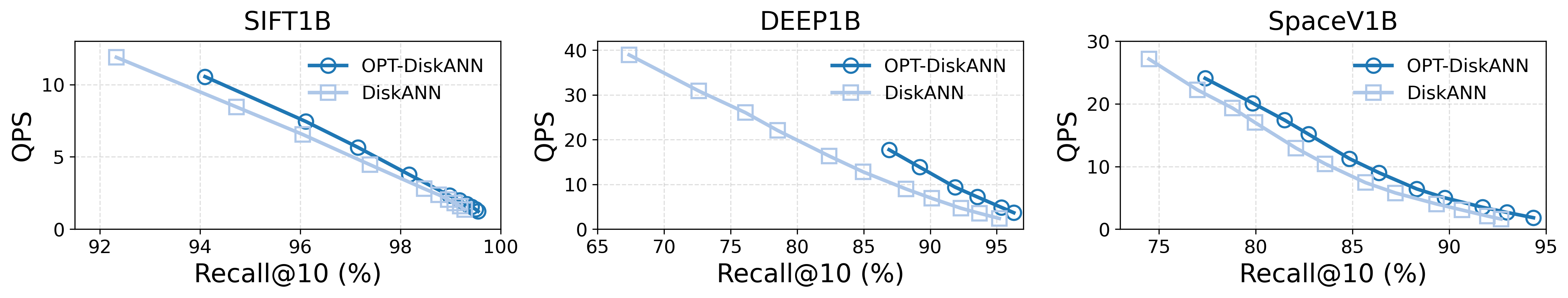}
    \vspace{-0.4cm}
    \caption{Query performance on billion datasets}
    \label{fig:billion-sets}
\end{figure}

\raggedbottom

\textbf{Approximate Nearest Neighbor Search.}
ANNS has evolved to support diverse scenarios, including multi-metric spaces
\cite{Multi}, hybrid queries \cite{ACORN}, privacy-preserving federated settings
\cite{FEDNNsig2024}, and retrieval-augmented generation \cite{RAG}. Existing
methods are commonly categorized into tree-based \cite{tree1,tree2},
hashing-based \cite{LSHarshing,hash2,hash3}, quantization-based
\cite{quantti1}, and graph-based approaches \cite{Graph1,Graphprob,Graph3}.
Tree-based methods, such as \(k\)-d trees~\cite{KD-TREE}, hierarchically
partition the space but suffer from the curse of dimensionality. Hashing-based
methods, such as LSH~\cite{LSHarshing}, map similar points to the same bucket
with high probability but often incur substantial memory overhead.
Quantization-based methods, including product quantization (PQ)
\cite{DBLP:journals/pami/JegouDS11} and optimized product quantization (OPQ)
\cite{OPQ}, compress vectors for efficient distance computation at the cost of
precision loss. Hybrid systems combining graph search with quantization, such
as DiskANN \cite{NEURIPS2019} and SPANN \cite{Graph2024VLDB}, further support
billion-scale search.

\textbf{Graph-Based ANNS Framework.}
Graph-based ANNS systems construct sparse proximity graphs to balance search
efficiency and index size. Representative systems include HNSW \cite{HNSW2020},
NSG \cite{NSG19}, KGraph~\cite{Kgraph}, DiskANN~\cite{NEURIPS2019},
DPG~\cite{DPG}, and SPTAG~\cite{SPANN2021}. Extensions such as
Filtered-DiskANN~\cite{gollapudi2023filtered} and FreshDiskANN~\cite{FRESH}
adapt this framework to filtered or dynamic large-scale settings, while recent
work explores cross-modal retrieval~\cite{chen2024roargraph} and filtered
search~\cite{ACORN,ait2025rwalks}.

Existing theoretical studies mainly analyze structural properties such as degree
bounds, query complexity, and search path length under simplified distributional
assumptions. Laarhoven~\cite{ICML2020} establishes sublinear query complexity
for threshold-based nearest neighbor graphs under spherical data models, where
sufficient local connectivity enables efficient greedy routing. Prokhorenkova
and Shekhovtsov~\cite{Haya24} analyze navigable graphs under a binomial
distribution and show that average degree \(O(n^\beta)\) for \(\beta<1/2\) can
preserve navigability, revealing a trade-off between graph sparsity and search
reliability.

\textbf{Parameter Tuning for ANNS Systems.}
Selecting optimal parameters for ANNS systems remains challenging. Traditional
parameter sweeping requires repeated index construction and evaluation, which is
computationally expensive. Automated approaches such as VDTuner~\cite{vdtuner}
use learning-based models to predict configurations. Our work takes a
complementary route by deriving analytical parameter-selection rules from the
pruning dynamics, thereby reducing the need for iterative search.
\section{Conclusions}
\label{sec:conclusion}
This paper revisits SNG-based graph indexing for approximate nearest neighbor
search from theoretical and optimization perspectives. We develop OPT-SNG, a
martingale-based probabilistic framework for modeling the stochastic pruning
dynamics in SNG construction. This framework enables a principled analysis of
graph structures generated by widely used systems. We show that the constructed
SNG has maximum out-degree \(O(n^{2/3+\epsilon})\) for any \(\epsilon>0\), and
that \textit{GreedySearch} converges in an expected \(O(\log n)\) steps, giving
an expected-case, cardinality-based complement to existing analyses.

Building on these guarantees, we derive a closed-form rule for selecting the
truncation parameter \(R\) as a function of \(\alpha\) and \(n\), together with
a lightweight calibration procedure for estimating the scaling constant.
Experiments on real-world benchmarks and a synthetic UNIFORM dataset show that
our analytical parameter selection removes expensive parameter sweeping,
reducing construction time by \(5.9\times\) on average while maintaining or
improving search performance.

\rev{An important direction for future work is to extend
Theorem~\ref{theorem:path} beyond the current uniform-distribution setting. A
rigorous extension would require re-establishing the volume-ratio inequalities
and drift conditions used in the proof under non-uniform distributions.}
\revnote{R1.O3}
\balance
\newpage
\bibliographystyle{ACM-Reference-Format}
\bibliography{sample-base}

\clearpage
\appendix

\section{Supplementary Lemmas}
\label{app:sup}
We first introduce a commonly used lemma in probability theory, which determines the probability of the limsup of events by analyzing the convergence of the series.
\label{app:lemma}
\begin{lemma}[Borel–Cantelli lemma]
\label{lem:BC}
Let $(\Omega, \mathcal{F}, \mathbb{P})$ be a probability space, where:
 $\Omega$ is the set of all possible outcomes, $\mathcal{F}$ is a sigma-algebra that contains all the information about events, and $\mathbb{P}$ is the associated probability measure.
Given a sequence of events $A_1, A_2, \ldots \in \mathcal{F}$ such that $\sum_{n=1}^{\infty} \mathbb{P}(A_n) < \infty$, it follows that:
$$
\mathbb{P}\left(\bigcap_{N=1}^{\infty}\left(\bigcup_{n \geq N} A_n\right)\right) = 0.
$$
Here, $\bigcap_{N=1}^{\infty}\left(\bigcup_{n \geq N} A_n\right)$ is also commonly denoted as $\limsup A_n$.
\end{lemma}
The following Chernoff bound provides a probabilistic guarantee on the deviation of a binomial random variable from its expected value.
\begin{lemma}[Chernoff Bound of binomial variable]
\label{lem:cheby}
Let $X \sim \operatorname{Bin}(n, p)$ and let $\mu=\mathbb{E}[X]$. For any $0<\delta<1$:\\
$$
\mathbb{P}(X \leq(1-\delta) \mu) \leq \exp \left(-\frac{\delta^2 \mu}{2}\right)
.$$
\end{lemma}
The frame of Wormald’s differential equation method (DEM) was developed by Wormald in 1990s \cite{DEM1} as a powerful tool for analyzing the discrete-time randomized graph processes and algorithms. Given a discrete-time stochastic process, in particular, consider a submartingale, if the one-step differences are bounded and the expected differences are well-approximated by a Lipschitz function, it follows that the process closely aligns with the corresponding differential equation with high probability.
\begin{lemma}[Differential Equation method \cite{DEM1}] 
\label{lem DEM}
Given integers \(n \geq 1\), a bounded domain \(\mathcal{D} \subseteq \mathbb{R}^{a+1}\), a function \(F: \mathcal{D} \to \mathbb{R}\), and a sequence of sigma algebras \(\mathcal{F}_0 \subseteq \mathcal{F}_1 \subseteq \cdots\). Define random variables \(Y(i)\) such that \(Y(i)\) is \(\mathcal{F}_i\)-measurable for \(i \geq 0\).

Assume that for all \(i \geq 0\), when \((i / n, Y(i) / n) \in \mathcal{D}\), the following conditions are satisfied. Additionally, let \(\mathcal{D}\) contain the closure of the set \(\{(0, z_0): \mathbb{P}(Y(0) = z_0 n) \neq 0\}\) for some \(n\), meaning that \(\mathcal{D}\) includes all possible starting points. 

If \(Y(t)\) is a submartingale and the following three conditions hold:

\noindent
\textbf{(1) (Boundedness)} For some \(\beta = \beta(n) \geq 1\) and \(\gamma = \gamma(n)\), the probability that \(|Y(i+1) - Y(i)| \leq \beta\) given \(\mathcal{F}_i\) is at least \(1 - \gamma\). This ensures that the change in \(Y(i)\) over one step is controlled.\\
\noindent
\textbf{(2) (Trend)}  For some \(\lambda_1 = \lambda_1(n) = o(1)\), the gap between the conditional expectation of the change and the normalized value of the function \(F\) is insignificant:  
   \[
   \left| \mathbb{E}(Y(i+1) - Y(i) \mid \mathcal{F}_i) - F\left(\frac{i}{n}, \frac{Y(i)}{n}\right) \right| \leq \lambda_1.
   \]\\
\noindent
\textbf{(3) (Lipschitz) }\(F\) is continuous and Lipschitz with constant \(L\) on \(D\).\\
\noindent
Then we have:\\
\noindent
(a) For \((0, \hat{z}_0) \in \mathcal{D}\), the differential equation  
\[
\frac{dz}{dx} = F(x, z)
\]  
has a unique solution in \(\mathcal{D}\) with the initial condition \(z(0) = \hat{z}_0\). This solution can be extended close to the boundary of \(\mathcal{D}\).\\
\noindent
Let \(\lambda > \lambda_1 + C_0 n \gamma\), where \(\lambda = o(1)\). Then, with probability  
\[
1 - O\left(n \gamma + \frac{\beta}{\lambda} \exp\left(-\frac{n \lambda^3}{\beta^3}\right)\right),
\]  
the following holds:  
\[
Y(i) = n z\left(\frac{i}{n}\right) + O(\lambda n),
\]  
where \(z(x)\) is the solution from part (a) with the initial condition \(\hat{z}_0 = \frac{1}{n} Y(0)\). This result is valid uniformly for \(0 \leq i \leq \sigma n\), where \(\sigma = \sigma(n)\) is the supremum of \(x\) such that the solution \(z(x)\) remains within \(\mathcal{D}\).
\end{lemma}
We will also use Jensen's inequality in the proof. It establishes a fundamental relationship between the expectation of a convex transformation of a random variable and the transformation of its expectation.
\begin{lemma}[Jensen's Inequality]
\label{lem:Jensen}
Suppose $g$ is convex and $X$ and $g(X)$ are both integrable. Then
$$
g(\mathbb{E} X) \leq \mathbb{E} g(X)
$$
\end{lemma}

\section{Omitted Proofs}\label{app:proofs}
\subsection{Proof of Lemma~\ref{lem:probability}}
\label{app:probability}
The following is the complete proof of Lemma~\ref{lem:probability}.
\begin{proof}[Proof of Lemma~\ref{lem:probability}]
\begin{figure}[H]
    \centering
\includegraphics[width=0.45\linewidth]{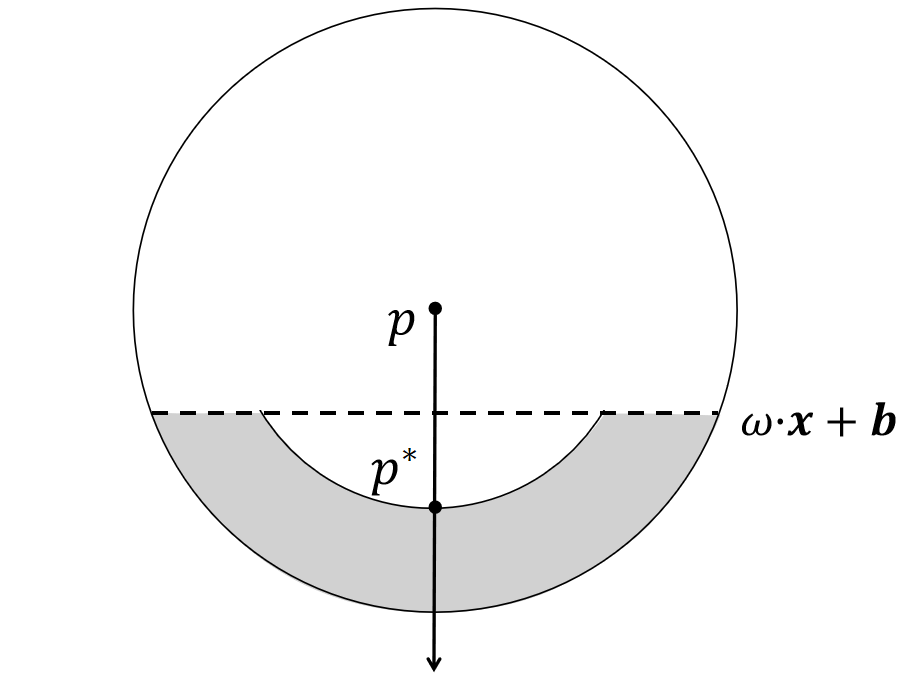}
    \caption{Given that $p^{*}$ is the nearest point to $p$ in $t$-th iteration and $\|p-p^*\|=\rho_t$. The light grey region represents all \(p'\) satisfied pruning condition \( \|p-p'\| > \|p^*- p'\| \).}
    \label{fig:prob}
\end{figure}
Assuming a uniform distribution over the search space, the probability that a randomly sampled point falls within a region is proportional to the region's volume. Without loss of generality, we place the center point $p$ at the origin and align the vector $\overrightarrow{pp^*}$ with the last coordinate axis, as shown in Figure \ref{fig:prob}. Under this coordinate system, the nearest neighbor $p^*$ lies at the position $(0, 0, \dots, 0, \rho_t) \in \mathbb{R}^d$. Let $p^{*}$ denote the nearest point to $p$ in $t$-th iteration and $\|p-p^*\|=\rho_t$, then the hyperplane is defined by:
$$
\omega \cdot x + b = 0, \quad \text{where } \omega = (0, \dots, 0, 1), \quad b = -\frac{\rho_t}{2}.
$$
The pruned region at iteration \(t\), denoted by $D(t)$, is defined as the intersection of the ball $B(p,\rho_0)$  and the half-space $\omega \cdot x > \frac{\rho_t}{2}$.
Let $V(\cdot)$ denote the Lebesgue volume in the appropriate dimension. For example, $V(B_d(r))$ is the volume of a $d$-dimensional ball of radius $r$, and $V(B_{d-1}(r))$ is the volume of a $(d-1)$-dimensional ball of radius $r$. In particular, $V(B_d(1))$ denotes the volume of the $d$-dimensional unit ball:
$$
V(B_d(1)) = \frac{\pi^{d/2}}{\Gamma(1 + d/2)}.
$$
To compute $V(D(t))$, we slice the ball orthogonally along the last coordinate $x_d$, and integrate over the $(d-1)$-dimensional cross-sectional volumes. The volume of the pruned region is given by:
$$
V(D(t)) = \int_{x_d = \frac{\rho_t}{2}}^{\rho_0} V\left( B_{d-1}\left( \sqrt{\rho_0^2 - x_d^2} \right) \right) \, dx_d,
$$
where $B_{d-1}(\sqrt{\rho_0^2 - x_d^2})$ is the $(d-1)$-dimensional ball of radius $\sqrt{\rho_0^2 - x_d^2}$, corresponding to the horizontal slice at height $x_d$.
Using the change of variables $t = x_d / \rho_0$, we have $x_d = \rho_0 t$ and $dx_d = \rho_0 \, dt$. Substituting this into the integral yields:
$$
V(D(t)) = V(B_{d-1}(1)) \cdot \rho_0^d \int_{\frac{\rho_t}{2\rho_0}}^1 (1 - t^2)^{\frac{d-1}{2}} \, dt.
$$
Next, we apply the substitution $z = 1 - t^2$ to convert the integral into a standard form. Under this change of variables, we obtain:
$$
\int_{\frac{\rho_t}{2\rho_0}}^1 (1 - t^2)^{\frac{d-1}{2}} \, dt = \frac{1}{2} \int_0^{1 - \left( \frac{\rho_t}{2\rho_0} \right)^2} z^{\frac{d-1}{2}} (1 - z)^{-1/2} \, dz.
$$
This integral corresponds to the incomplete Beta function. After normalization, it becomes proportional to the \textit{regularized incomplete Beta function}, also known as the cumulative distribution function of the Beta distribution, denoted by $I_x(a, b)$. Thus, we have:
$$
\int_{\frac{\rho_t}{2\rho_0}}^1 (1 - t^2)^{\frac{d-1}{2}} \, dt = \frac{\Gamma\left( \frac{d+1}{2} \right)\Gamma\left( \frac{1}{2} \right)}{2\Gamma\left( \frac{d}{2} + 1 \right)} \cdot I_{1 - \left( \frac{\rho_t}{2\rho_0} \right)^2} \left( \frac{d+1}{2}, \frac{1}{2} \right),
$$
where $\Gamma(\cdot)$ denotes the Gamma function. Hence, the volume of the pruned region is:
$$
V(D(t)) =\frac{\Gamma(\frac{d+1}{2})\Gamma(\frac{1}{2})}{2\Gamma(\frac{d}{2}+1)} V(B_{d-1}(1)) \cdot \rho_0^d \cdot I_{1 - (\frac{\rho_t}{2\rho_0})^2} \left( \frac{d+1}{2}, \frac{1}{2} \right).
$$
However, since this includes points within the inner ball $B(\rho_t)$, and we subtract the volume within $B(\rho_t)$, denoted \(V_{\text{inner}}\), is: 
$$
V_{\text{inner}} = \frac{\Gamma(\frac{d+1}{2})\Gamma(\frac{1}{2})}{2\Gamma(\frac{d}{2}+1)}V(B_{d-1}(1)) \cdot \rho_t^d \cdot I_{3/4} \left( \frac{d+1}{2}, \frac{1}{2} \right),
$$
where $3/4 = 1 - (1/2)^2$ accounts for the halfway-positioned hyperplane. The probability that a randomly sampled point is pruned in iteration $t$ is:
$$
\pi_t = \frac{V(D(t)) - V_{\text{inner}}}{V(B_d(\rho_0)) - V(B_d(\rho_t))}.
$$
Substituting the expressions above and simplifying using $V(B_d(1)) = \frac{\pi^{d/2}}{\Gamma(\frac{d}{2}+1)}$, we get:
$$
\pi_t = \frac{ I_{1 - (\frac{\rho_t}{2\rho_0})^2} - (\frac{\rho_t}{\rho_0})^d \cdot I_{3/4} }{2 \left(1 - (\frac{\rho_t}{\rho_0})^d\right)}.
$$
Since $\frac{\rho_t}{\rho_0} \in (0, 1)$ and $1 - (\frac{\rho_t}{2\rho_0})^2 \in (3/4, 1)$, the ratio is always decreases with respect to \(\frac{\rho_t}{\rho_0}\). Specifically, noting that $I_{3/4}(a, b)  < 1$, for large \(d\), we derive:
$$
\frac{1}{2} I_{3/4} \left( \frac{d+1}{2}, \frac{1}{2} \right) < \pi_t < \frac{1}{2}.
$$
\textbf{2-dimensional case.} In the two-dimensional case, we can leverage geometric relationships to derive a lower bound of $\frac{1}{3}$ for the pruning probability. This bound corresponds to the scenario where the pruning probability decreases as $\rho_t/\rho_0$ increases, reaching its minimum when $\rho_t / \rho_0 = 1$—that is, when the distance from the center point $p$ to its nearest neighbor $p^*$ approaches the radius of the disk. In this situation, the pruning decision boundary can be approximated by the perpendicular bisector of the segment connecting \(p\) and \(p^*\). Geometrically, this bisector divides the circular region into two sectors, and the prunable region corresponds to the sector opposite $p$. As the angle subtended by this region exceeds $\frac{2\pi}{3}$, it follows that the fraction of the total area being pruned is at least $\frac{1}{3}$. Thus, the pruning probability is bounded below by $\frac{1}{3}$ in the two-dimensional setting.
\end{proof}

\subsection{Conservative Probability Estimation}
\label{app:Applicability}
\textbf{Applicability of Pruning Probability to General Points.} Our probability calculations are specifically designed for the scenarios where the reference point is located at the center of the sphere, ensuring uniform density in all directions. Under the assumption of a uniform distribution, each point in the dataset can be locally approximated as having uniform density in its neighborhood. Consequently, the pruning probability derived in Lemma~\ref{lem:probability} extends naturally to arbitrary points in the dataset, justifying its applicability throughout the SNG construction process.

\textbf{Overlap.} We also note that in later iterations, the pruning region at step $t$ may overlap with those from earlier steps. This geometric overlap does not invalidate our degree bound—in fact, it tends to have faster elimination of candidates and thus requires fewer steps (i.e., fewer edges). In the proof of Lemma~\ref{lem:probability}, any point that lies within the overlapping regions of multiple pruning areas remains eligible for removal, effectively undergoing repeated pruning. Therefore, our earlier volume-based estimate, which assumes disjoint pruning regions and ignores such overlaps, is conservative.

\subsection{Proof of Theorem \ref{theorem:max_degree}}
\label{app:max_degree}
\begin{proof}[Proof of Theorem~\ref{theorem:max_degree}]
Let \(T\) denote the total number of pruning iterations in the SNG construction.
By Lemma~\ref{lemma:prop}(3), the maximum out-degree of the constructed SNG is
controlled by \(T\). Therefore, it suffices to prove that, for any fixed
\(\varepsilon>0\),
\[
T=O(n^{2/3+\varepsilon})
\]
with probability \(1\).

The proof has two main ingredients. The first is a uniform positive lower bound
on the per-iteration pruning probability:
\[
\pi_t\ge c_*>0.
\]
The second is the two-phase pruning analysis developed in
Sections~\ref{subsec:2-phase}--\ref{subsec:highdim}, which combines the
Chernoff bound, the Borel--Cantelli argument, and Wormald's differential
equation method (DEM). Since this two-phase analysis only requires the existence
of a positive lower bound on \(\pi_t\), the proof is insensitive to the specific
distributional model once such a constant \(c_*\) has been established. We
therefore first prove the result in the uniform \(\alpha=1\) setting, and then
extend the same argument to general \(\alpha>0\) and to bounded non-uniform
densities.

\medskip
\noindent\textbf{Case~1: uniform distribution and \(\alpha=1\).}
We first consider the base setting where the candidates are independently and
uniformly distributed in a \(d\)-dimensional ball \(B(p,\rho_0)\), and
\(\alpha=1\). In this case, the pruning condition for a candidate \(x\) with
respect to the selected nearest neighbor \(p_t^*\) is
\[
\|p-x\|\ge \|p_t^*-x\|.
\]
Equivalently, the pruned region is a spherical cap cut out by the perpendicular
bisector of the segment joining \(p\) and \(p_t^*\). By
Lemma~\ref{lem:probability}, the normalized volume of this region is bounded
below by a positive constant depending only on \(d\). Hence there exists
\[
c(d,1)>0
\]
such that, for every iteration \(t\),
\[
\pi_t\ge c(d,1).
\]

Let \(S'_t\) denote the set of processed points after \(t\) iterations, and set
\[
Y_t:=|S'_t|.
\]
Since the unprocessed candidate set has size
\[
|S_t|=n-1-Y_t,
\]
the one-step increment satisfies
\[
Y_{t+1}-Y_t=1+\Delta S_{t+1},
\]
where the term \(1\) corresponds to the selected nearest candidate, and
\(\Delta S_{t+1}\) is the number of additionally pruned candidates. Conditional
on the natural filtration \(\mathcal{F}_t\), the conditional expectation is
bounded below by
\[
\mathbb{E}[Y_{t+1}-Y_t\mid \mathcal{F}_t]
=
1+\mathbb{E}[\Delta S_{t+1}\mid \mathcal{F}_t]
\ge
1+c(d,1)(n-1-Y_t).
\]
After normalization \(Z_t:=Y_t/n\), this becomes
\[
\mathbb{E}[Z_{t+1}-Z_t\mid \mathcal{F}_t]
\ge
\frac{1+c(d,1)(n-1-nZ_t)}{n}.
\]
Thus the associated deterministic drift is, up to lower-order terms,
\[
F(z)=1+c(d,1)(1-z).
\]

We now split the process into two phases. Fix any
\(\nu\in(2/3,1)\). The first phase lasts until the processed set reaches size
\(n-n^{1-\nu}\), and the second phase covers the remaining iterations:
\[
\text{Phase I: } |S'_t|<n-n^{1-\nu},
\qquad
\text{Phase II: } |S'_t|\ge n-n^{1-\nu}.
\]

\smallskip
\noindent\emph{First phase.}
By Lemma~\ref{lem:sublinear}(ii), for any \(\nu\in(0,1)\), the SNG construction
reaches the \((n-n^{1-\nu})\)-level within
\[
t_1=O(n^\nu)
\]
iterations with probability \(1\). In other words, after \(O(n^\nu)\)
iterations, all but at most \(n^{1-\nu}\) candidates have already been processed
or pruned.

\smallskip
\noindent\emph{Second phase.}
It remains to bound the number of iterations needed after \(t_1\). In this
plateau phase, the remaining candidate set satisfies
\[
|S_t|=n-1-|S'_t|\le n^{1-\nu}.
\]
Let
\[
Y(i):=|S'_{t_1+i}|
\]
be the processed-set size \(i\) steps after the start of the second phase. We
verify the three DEM conditions for this residual process.

\smallskip
\noindent\emph{Boundedness.}
For every \(i\ge 0\),
\[
|Y(i+1)-Y(i)|
=
1+\Delta S_{t_1+i+1}
\le |S_{t_1+i}|
\le n^{1-\nu}.
\]
Thus the one-step change is deterministically bounded by
\[
\beta:=n^{1-\nu},
\]
with exception probability \(\gamma=0\).

\smallskip
\noindent\emph{Trend.}
The conditional expectation of one increment is
\[
\mathbb{E}[Y(i+1)-Y(i)\mid \mathcal{F}_i]
=
1+\pi_{t_1+i+1}\bigl(|S_{t_1+i}|-1\bigr).
\]
Writing \(z=Y(i)/n\), we have
\[
|S_{t_1+i}|=n-1-Y(i)=n(1-z)-1.
\]
Using the lower bound \(\pi_t\ge c(d,1)\), the normalized drift is bounded below
by
\[
\mathbb{E}\!\left[\frac{Y(i+1)-Y(i)}{n}\,\middle|\,\mathcal{F}_i\right]
\ge
\frac{1+c(d,1)(n(1-z)-2)}{n}.
\]
Equivalently, up to an \(o(1)/n\) error, the drift is governed by
\[
F(z)=1+c(d,1)(1-z).
\]
That is,
\[
\mathbb{E}[Z_{i+1}-Z_i\mid\mathcal{F}_i]
=
\frac{F(Z_i)+o(1)}{n}.
\]

\smallskip
\noindent\emph{Lipschitz continuity.}
For any \(z_1,z_2\in[0,1]\),
\[
|F(z_1)-F(z_2)|
=
c(d,1)|z_1-z_2|.
\]
Hence \(F\) is Lipschitz on the relevant domain.

Therefore, Wormald's differential equation method applies. The corresponding
differential equation is
\[
\frac{dz}{dx}=1+c(d,1)(1-z),
\qquad
z(0)=1-n^{-\nu},
\]
where \(x=i/n\) is the rescaled time. Solving this ODE gives
\[
z(x)
=
1+\frac{1}{c(d,1)}
-
\left(\frac{1}{c(d,1)}+n^{-\nu}\right)e^{-c(d,1)x}.
\]
To estimate the time required to reach \(z=1-1/n\), set
\[
z(x)=1-\frac{1}{n}.
\]
Then
\[
1-\frac{1}{n}
=
1+\frac{1}{c(d,1)}
-
\left(\frac{1}{c(d,1)}+n^{-\nu}\right)e^{-c(d,1)x},
\]
which implies
\[
e^{-c(d,1)x}
=
\frac{\frac{1}{c(d,1)}+\frac{1}{n}}
     {\frac{1}{c(d,1)}+n^{-\nu}}.
\]
Thus
\[
x
=
\frac{1}{c(d,1)}
\log\!\left(
\frac{\frac{1}{c(d,1)}+n^{-\nu}}
     {\frac{1}{c(d,1)}+\frac{1}{n}}
\right).
\]
Since \(n^{-\nu}\to 0\), using \(\log(1+u)=O(u)\) for small \(u\), we obtain
\[
x=O(n^{-\nu}).
\]
Because the actual number of iterations is \(i=nx\), the second phase takes
\[
T_2=O(n^{1-\nu})
\]
iterations.

The DEM approximation holds with probability tending to one. Combined with the
Borel--Cantelli argument used in Lemma~\ref{lem:sublinear}, the above two-phase
bound holds with probability \(1\). Consequently,
\[
T=t_1+T_2
=
O(n^\nu)+O(n^{1-\nu})
=
O(n^\nu).
\]
Since \(\nu\in(2/3,1)\) is arbitrary, taking
\[
\nu=2/3+\varepsilon
\]
yields
\[
T=O(n^{2/3+\varepsilon})
\]
with probability \(1\).

\medskip
\noindent\textbf{Case~2: uniform distribution and general \(\alpha>0\).}
We now consider the uniform model with a general fixed relaxation parameter
\(\alpha>0\). The pruning condition becomes
\[
\|p-x\|\ge \alpha\|p_t^*-x\|.
\]
For \(\alpha=1\), this reduces to the spherical-cap region in Case~1. For
\(\alpha\ne 1\), the boundary
\[
\|p-x\|=\alpha\|p_t^*-x\|
\]
is an Apollonius sphere. Hence the pruning region is described by the
intersection of the current support with the corresponding Apollonius-type
blocking region.

Although the geometry differs from the hyperplane cut in the \(\alpha=1\) case,
the crucial property remains the same: for fixed \(d\) and fixed
\(\alpha>0\), the normalized volume of the blocking region is bounded below by
a positive constant. Therefore, there exists
\[
c(d,\alpha)>0
\]
such that
\[
\pi_t\ge c(d,\alpha)
\]
uniformly over the iterations.

\textbf{The two-phase proof in Case~1 uses the geometry only through the positive
constant appearing in the drift lower bound. Hence, after establishing the
new lower bound \(c(d,\alpha)\), all remaining steps are identical to those in
Case~1, with \(c(d,1)\) replaced by \(c(d,\alpha)\).} The drift function becomes
\[
F_\alpha(z)=1+c(d,\alpha)(1-z).
\]
The boundedness condition is unchanged, the trend condition follows from
\(\pi_t\ge c(d,\alpha)\), and the Lipschitz condition holds because
\[
|F_\alpha(z_1)-F_\alpha(z_2)|
=
c(d,\alpha)|z_1-z_2|.
\]
Therefore, by the same Chernoff, Borel--Cantelli, and DEM argument,
\[
T=O(n^{2/3+\varepsilon})
\]
with probability \(1\) for every fixed \(\alpha>0\) under the uniform
distribution model.

\medskip
\noindent\textbf{Case~3: bounded non-uniform density.}
Finally, suppose that the local candidate distribution is not uniform but has a
density \(f\) on the current support \(\Omega_{t-1}\), where
\[
0<m\le f(x)\le M<\infty
\]
for all \(x\in\Omega_{t-1}\). Let \(R_t(p_t^*)\) denote the blocking region
defined by
\[
R_t(p_t^*)
:=
\{x\in \Omega_{t-1}:\|p-x\|\ge \alpha\|p_t^*-x\|\}.
\]
Conditional on \(p_t^*\), the pruning probability under the density \(f\) is
\[
\pi_t(\mu\mid p_t^*)
=
\frac{\int_{R_t(p_t^*)} f(x)\,dx}
     {\int_{\Omega_{t-1}} f(x)\,dx}.
\]
Using \(m\le f\le M\), we obtain
\[
\pi_t(\mu\mid p_t^*)
\ge
\frac{m\,\operatorname{Vol}(R_t(p_t^*))}
     {M\,\operatorname{Vol}(\Omega_{t-1})}
=
\frac{m}{M}
\frac{\operatorname{Vol}(R_t(p_t^*))}
     {\operatorname{Vol}(\Omega_{t-1})}.
\]
By the uniform-volume lower bound from Case~2,
\[
\frac{\operatorname{Vol}(R_t(p_t^*))}
     {\operatorname{Vol}(\Omega_{t-1})}
\ge c(d,\alpha).
\]
Hence
\[
\pi_t(\mu\mid p_t^*)
\ge
\frac{m}{M}c(d,\alpha)
=:c_*>0.
\]

Thus the bounded non-uniform density changes only the constant in the pruning
probability lower bound. The drift function becomes
\[
F_*(z)=1+c_*(1-z).
\]
As before, the boundedness condition is unchanged, the trend condition follows
from \(\pi_t\ge c_*\), and \(F_*\) is Lipschitz since
\[
|F_*(z_1)-F_*(z_2)|
=
c_*|z_1-z_2|.
\]
Therefore the same two-phase pruning argument applies. We again obtain
\[
T=O(n^{2/3+\varepsilon})
\]
with probability \(1\).

Combining the three cases, the total number of pruning iterations satisfies
\[
T=O(n^{2/3+\varepsilon})
\]
with probability \(1\) under the stated assumptions. 
\end{proof}

\subsection{Full Proof to Theorem \ref{theorem:path}}
\label{app:path}
\begin{proof}[Proof of Theorem \ref{theorem:path}]
We analyze the expected length \(k\) of a \textit{GreedySearch} path from \(p\) to the query \(v_k\) \( \{p=v_0, v_1, \ldots, v_k\} \), where each node selects its neighbor closest to the query point. Let \(\eta\) denote the out-degree (number of neighbors) of \(v_k\), and assume that points in \(P\) are independently and uniformly distributed in a \(d\)-dimensional ball of radius \(\rho_0\).
\begin{figure}[H]
    \centering
    \includegraphics[width=0.3\linewidth]{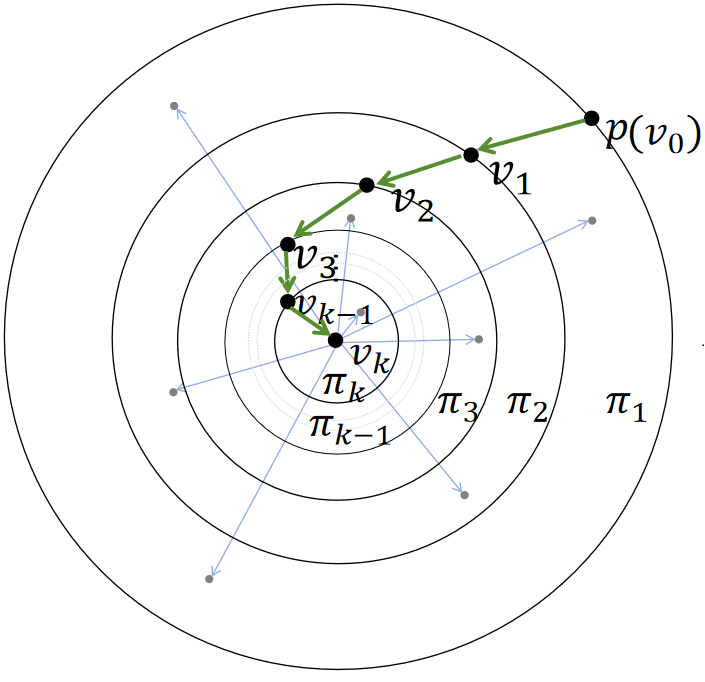}
    \caption{Path and neighbor distribution diagram: Blue lines represent edges from \( v_k \) to its neighbors, and green lines represent the search path. \( \pi_i \) denotes the number of neighbors in the \( i \)-th annular layer from the outermost to the innermost.}
    \label{fig:neighbor_dist}
\end{figure}
Define \(\triangle r\) as the minimum distinguishable pairwise distance difference among all point triples: 
\begin{align*}
\triangle r = \min \{ |\|a - b\| - \|a - c\||, |\|a - b\| - \|b - c\||, |\|a - c\| - \|b - c\||
\\ \mid a, b, c \in P, \text{distinct} \}    
\end{align*}
as introduced in \cite{NSG19}.
Let \(v_k\) be the final node of the search path. We construct concentric balls centered at \(v_k\) with radii \(\|v_k-v_i\|\) for \(i=0,1,...,k-1\), and define the annular layer  \(A_i= B(v_k, \|v_k-v_{i-1}\|) \setminus B(v_k, \|v_k - v_i\|) \). Let \(\eta_i\) be the number of neighbors in each layer \(A_i\). Then \(\eta_i\sim \text{Bin}\left( \eta, \frac{V(A_i)}{V(B(v_k, \rho_0))} \right)\), and the expected number in each layer is proportional to the volume difference:
\[
\mathbb{E}[\eta_i] = \eta \cdot \frac{V(B(v_k, \|v_k - v_{i-1}\|)) - V(B(v_k, \|v_k - v_i\|))}{V(B(v_k, \rho_0))}.
\]
 Summing over all \(k\) layers, we have:
 \begin{align}
 E\left[ \sum_{i=1}^k \eta \cdot \frac{V(B(v_k, \|v_k - v_{i-1}\|)) - V(B(v_k, \|v_k - v_i\|))}{V(B(v_k, R))} \right] \\\leq n \cdot \frac{V(B(v_k, \|p - v_k\|))}{V(B(v_k, R))}.    
 \end{align}

This upper bound is based on the fact that all points in the ball may serve as neighbors. Using $V(B(v\_k, r)) \propto r^d$, we obtain:
\begin{equation}
\label{equa:vol}
    E\left[ \sum_{i=1}^k \eta \cdot \frac{\|v_k - v_{i-1}\|^d - \|v_k - v_i\|^d}{\rho_0^d} \right] \leq n \cdot E\left[ \frac{\|p - v_k\|^d}{\rho_0^d} \right].
\end{equation}
 Now, note that
 \begin{align}
&\|v_k - v_{i-1}\|^d - \|v_k - v_i\|^d 
\\&= (\|v_k - v_{i-1}\| - \|v_k - v_i\|) \sum_{j=0}^{d-1} \|v_k - v_{i-1}\|^j \|v_k - v_i\|^{d-1-j}.     
 \end{align}

Since $\|v_k - v_{i-1}\| - \|v_k - v_i\| \geq \triangle r$ and $\|v_k - v_i\| \geq \|v_k - v_{k-1}\|$ for all $i$, we have:
$$
\|v_k - v_{i-1}\|^d - \|v_k - v_i\|^d \geq d \cdot \triangle r \cdot \|v_k - v_{k-1}\|^{d-1}.
$$
Substituting into Equation~\ref{equa:vol}, we obtain:
$$
\mathbb{E}[k \cdot \eta \cdot d \cdot \triangle r \cdot \|v_k - v_{k-1}\|^{d-1}] \leq n \cdot \mathbb{E}[\|p - v_k\|^d].
$$
Assuming independence between path length and distances, we rearrange:
\begin{equation}
\label{eq:keta}
E[k \cdot \eta] \leq n \cdot \frac{E[\|p - v_k\|^d]}{\triangle r \cdot d\cdot E[\|v_k - v_{k-1}\|^{d-1}]}.
\end{equation}
Next, we upper bound $\|p - v_k\|$. Since the search path is monotonic:
\begin{equation}
\|p - v_k\| = \|v_0 - v_k\| \leq \|v_0 - v_k\| \cdot \prod_{i=1}^{k-1} \frac{\|v_i - v_k\|}{\|v_{i-1} - v_k\|} \cdot \frac{\|p - v_k\|}{\|v_{k-1} - v_k\|}.    
\end{equation}
Because $\frac{\|v_i - v_k\|}{\|v_{i-1} - v_k\|} \leq \frac{\rho_0 - \triangle r}{\rho_0}$ and $\frac{\|p - v_k\|}{\|v_{k-1} - v_k\|} \leq \frac{\rho_0}{\triangle r}$, we have:
\begin{equation}
\|p - v_k\|^d \leq \frac{\rho_0^{2d}}{\triangle r^d} \cdot \left( \frac{\rho_0 - \triangle r}{\rho_0} \right)^{d(k-1)}.    
\end{equation}
Substituting this into Equation~\ref{eq:keta}, we obtain:
\begin{equation}
 \mathbb{E}[k \cdot \eta] \leq \frac{n \cdot \rho_0^{2d} \cdot \mathbb{E}\left[ \left( \frac{\rho_0 - \triangle r}{\rho_0} \right)^{d(k-1)} \right]}{d \cdot \triangle r^{d+1} \cdot \mathbb{E}[\|v_k - v_{k-1}\|^{d-1}]}.   
\end{equation}
Using the bound $\mathbb{E}[\|v_k - v_{k-1}\|] \geq \frac{\rho_0}{(n+1)^{1/d}}$ from \cite{Dist2010} and Jensen’s inequality(Lemma~\ref{lem:Jensen}), we find:
\begin{equation}
\mathbb{E}[\|v_k - v_{k-1}\|^{d-1}] \geq \left( \frac{\rho_0}{(n+1)^{1/d}} \right)^{d-1}.    
\end{equation}

Thus:
\begin{align}
    E[k \cdot \eta] &\leq \frac{n (n+1)^{(d-1)/d} \rho_0^{d+1} \cdot E\left[ \left( \frac{\rho_0 - \triangle r}{\rho_0} \right)^{d \cdot (k-1)} \right]}{\triangle r^{d+1} \cdot d}\\
&\leq \frac{n (n+1)^{(d-1)/d} \rho_0^{d+1} \cdot E\left[ -\left( \frac{\rho_0 - \triangle r}{\rho_0} \right)^{d \cdot (k-1)} +2\right]}{\triangle r^{d+1} \cdot d}.
\end{align}
Applying Jensen’s inequality again:
\begin{equation}
 E[k \cdot \eta] \leq \frac{n (n+1)^{(d-1)/d} \rho_0^{d+1} \cdot (-\left( \frac{\rho_0 - \triangle r}{\rho_0} \right)^{d \cdot \mathbb{E}[k-1]}+2)}{\triangle r^{d+1} \cdot d}   
\end{equation}
Define the function:
\begin{equation}
g(x) := \eta \cdot x - \frac{n(n+1)^{(d-1)/d} \rho_0^{d+1}}{d \cdot \triangle r^{d+1}} \left[ -\left( \frac{\rho_0 - \triangle r}{\rho_0} \right)^{d(x-1)} + 2 \right].    
\end{equation}
We have \(E(k)<0\). Evaluating:
\[
g(0) < 0,
\]
\begin{equation}
g'(x) = \eta + \frac{n (n+1)^{(d-1)/d} \rho_0^{d+1}}{\triangle r^{d+1}\cdot d} \cdot d \cdot \ln \left( \frac{R}{R - \triangle r} \right) \cdot \left( \frac{R - \triangle r}{R} \right)^{d \cdot x} > 0,    
\end{equation}
\begin{equation}
g(\log n) = \eta \log n + \frac{n (n+1)^{(d-1)/d} \rho_0^{d+1}}{\triangle r^{d+1} \cdot d} \left[  \left( \frac{R - \triangle r}{R} \right)^{d \cdot (\log n-1)} -2\right].    
\end{equation}
we have \( \left( \frac{\rho_0 - \triangle r}{\rho_0} \right)^{d \cdot \log n} = e^{d \cdot \log n \cdot \ln \frac{\rho_0 - \triangle r}{\rho_0}} = n^{C'\cdot d\cdot \ln \frac{\rho_0 - \triangle r}{\rho_0}} >2 \). 
Thus:
\[
g(\log n) > 0.
\]
Since \( g(x) \) is monotonically increasing and transitions from negative to positive, we conclude that \(g(\log n)>0\) and \(g(E(k))<0\). Hence,\(E(k)<\log n\), leading to:
\[
E[k] = O(\log n).
\]
\end{proof}

\section{Complexity Analysis Details}
\label{app:complexity}
Below is a detailed complexity analysis of each component in the SNG construction algorithm.
\noindent\textbf{Complexity of GreedySearch.}
The GreedySearch algorithm iteratively selects the closest neighbor to the query and follows a path to an approximate nearest neighbor, evaluating all neighbors at each step. According to \cite{NSG19}, the complexity of GreedySearch is determined by the product of the average path length and the average degree. Theorem \ref{theorem:path} establishes that the path length is \( O(\log n) \). Thus, for a graph with degree \( K \), the search complexity is \( O(K \cdot \log n) \).

\noindent\textbf{Complexity of Pruning.}
In the SNG pruning algorithm, the outer loop checks if the candidate set \(S\) is empty. If not, it adds an edge to the current nearest neighbor. The inner loop evaluates all remaining points in \(S\), determining whether to prune them based on the pruning rule. In the non-truncated SNG pruning process, the outer loop iterates at most \( O(n^{2/3 + \epsilon}) \) times, as established by Theorem \ref{theorem:max_degree}. The inner loop evaluates all points in the candidate set \( S \), excluding the nearest neighbor, checking up to \( n-2 \) points per iteration. This yields a complexity upper bound of \( O(n^{5/3 + \epsilon}) \), tighter than the \( O(n^2) \) estimate in \cite{NEURIPS2019}. For the truncated variant, SNG-Prune, the outer loop runs at most \( R \) times. The inner loop evaluates up to \( n-2 \) points, resulting in a worst-case complexity of \( O(R \cdot n) \).
 
\textbf{Total Complexity of Construction od SNG.}  To begin with, the graph is randomly initialized as an $R$-regular graph. $R$ is the maximum degree specified by the parameter sweep. To reduce the indexing time, for each point $p$ in the dataset, GreedySearch is first performed starting from the centroid $s$ to obtain the set of volunteer points, which are then used as candidate points for pruning, the complexity can be approximate by $O(R \cdot logn)$ and written as $C_1\cdot R\log n+b_1$ \cite{NSG19}. Next, we apply SNG-tPrune to prune the randomized graph. This process runs for at most $R$ iterations, and in each iteration, the number of points to be checked is approximately $O(R \log n)$. Considering the parameter $\alpha$, as $\alpha$ increases, the graph has higher degree, hence the complexity can be regarded as inversely proportional to $\alpha$. Therefore, the algorithm's complexity is estimated to be $O(R \cdot R \log n / \alpha)$ and can be expressed as $C_2(R \cdot R \log n / \alpha) + b_2$. To enhance GreedySearch convergence during the search phase, we add reverse edges for each neighbor of point \( p \). This may cause a degree overrun beyond \( R \). In such cases, SNG-tPrune is reapplied. For each point \( p \) with at most \( R \) possible neighbors, we add reverse edges and reconstruct the graph. The candidate set size does not exceed \( R + 1 \), and the inner loop iterates at most \( R \) times. A larger \(\alpha\) increases the number of neighbors, raising the probability of a degree overrun, which we model as a proportional relationship, so the complexity is expressed as \(C_3(\alpha \cdot R \cdot R^2) + b_3\).

\section{Discussion}
\label{app:discussion}
\textbf{Point Processes and Uniformity: A New Perspective on Assumptions.}
Dataset point distributions can often be modeled by a \textit{spatial point process}, which is a random process used to describe the distribution of points in a given space. Formally, a spatial point process is a measurable mapping from a probability space into the space of locally finite point configurations in $\mathbb{R}^d$.

As elaborated in~\cite{ICML2020}, a widely used model is the \textit{Poisson point process}, which assumes that the number of points in any bounded region follows a Poisson distribution and that point locations are independent and identically distributed. Crucially, given a fixed number of points within a bounded region, the conditional distribution of those points under a Poisson point process is uniform. This property provides a natural justification for analyzing nearest neighbor graph constructions under the uniform distribution assumption, as the conditioning step aligns with the typical setup in ANNS frameworks such as DiskANN and SNG.

\textbf{Future Research Directions.} An important avenue for future research is to theoretically characterize the behavior of the SNG algorithm under diverse data distributions. Such analysis is crucial to ensure the robustness and effectiveness of SNG across datasets with varying structural properties. Advancing this line of work would further enhance the adaptability and generalizability of the proposed probability method, making it more suitable for real-world applications involving complex, non-uniform, and multimodal data distributions.

\end{document}